%% file: main.tex
\documentclass[runningheads]{llncs}
\usepackage[hidelinks,breaklinks]{hyperref}
\usepackage[electronic,xcolor,makeidx=compatibility,quotation,notion]{knowledge}

\usepackage{graphicx,wrapfig,lipsum}
\usepackage{tikz}
\usepackage{todonotes}
\usepackage{xspace}
\usepackage{cancel}
\usepackage{float}
\usepackage{amssymb}
\usepackage{amsmath}
\usepackage{hyperref}
\usepackage{caption}
\usepackage{centernot}
\usepackage{mathtools}
\usepackage{relsize}

\input{macros.tex}

\input{knowledgeinfintaryquantifiers.tex}

\begin{document}

%
\title{First-Order logic and its "Infinitary Quantifier" Extensions over "Countable Words"}
\titlerunning{Extended FO}
%
\author{Bharat Adsul$^*$  \and
Saptarshi Sarkar$^*$  \and
A.V. Sreejith$^\dagger$}
\authorrunning{B. Adsul et al.}
%
\institute{$^*$IIT Bombay, India and $^\dagger$IIT Goa, India}

\maketitle              
\begin{abstract}
\input{abstract.tex}

\end{abstract}

\section{Introduction}
\label{sec:intro}
\input{intro.tex}
%
\input{prelims-small.tex}
\section{Small fragments of FO}
\label{sec:fo}
\input{fo.tex}
%
\input{fo-inf.tex}
%
%
%
\section{No Finite Basis Theorems}
\label{sec:nofinbasis}
\input{nofinbasis.tex}
 \bibliographystyle{splncs04}
 \bibliography{papers}
%
%


%

\end{document}

%% file: macros.tex
\makeatletter
\newcommand{\xMapsto}[2][]{\ext@arrow 0599{\Mapstofill@}{#1}{#2}}
\def\Mapstofill@{\arrowfill@{\Mapstochar\Relbar}\Relbar\Rightarrow}
\makeatother

\newcounter{todocounter}

\newcommand{\nat}{\mathbb{N}}
\newcommand{\al}{\alphabet}
\newcommand{\cw}[1]{{#1}^\ostar}
\newcommand{\cwq}[2]{{#1}^\ostar \! / \! \! \sim_{#2}}
\newcommand{\fwq}[2]{{#1}^* \! / \! \! \sim_{#2}}

\newrobustcmd{\gnl}[1]{\kl[\gnl]{\mathrm{gnlen}{(#1)}}}
\knowledge {\gnl}{notion}

\newrobustcmd{\ei}[1]{\kl[\ei]{\exists^{\infty_{#1}}}}
\knowledge {\ei}{notion}

\newrobustcmd{\ali}[1]{\kl[\ali]{\Delta_{#1}}}
\knowledge {\ali}{notion}

\newrobustcmd{\foione}[1]{\kl[\foione]{\mathrm{FO}^1[(\infty_j)_{j \leq #1}]}}
\knowledge {\foione}{notion}
\newrobustcmd{\foi}[1]{\kl[\foi]{\mathrm{FO}[(\infty_j)_{j \leq #1}]}}
\knowledge {\foi}{notion}
\newrobustcmd{\foinf}{\kl[\foinf]{\ensuremath{\mathrm{FO}[\infty]}}}
\knowledge {\foinf}{notion}
\newrobustcmd{\irank}[1]{\kl[\irank]{\mathcal{I}\text{-}\mathit{rank}({#1})}}
\knowledge {\irank}{notion}

\newcommand{\bpc}{\blockproduct}
\newcommand{\ogi}[1]{\gamma_{#1}}

\newrobustcmd{\focut}{\kl[\focut]{\ensuremath{\mathrm{FO[cut]}}}}
\knowledge {\focut}{notion}
\newrobustcmd{\wmso}{\kl[\wmso]{\ensuremath{\mathrm{WMSO}}}}
\knowledge {\wmso}{notion}

\newrobustcmd{\fo}{\kl[\fo]{\ensuremath{\mathrm{FO}}}}
\knowledge {\fo}{notion}

\newrobustcmd{\bc}{\kl[\bc]{B(\exists^*)}}
\knowledge {\bc}{notion}

\newcommand{\mso}{\ensuremath{\mathrm{MSO}}}
\newcommand{\fotwo}{\ensuremath{\mathrm{FO}^2}}
\newcommand{\foone}{\ensuremath{\mathrm{FO}^1}}
\newcommand{\fothree}{\ensuremath{\mathrm{FO}^3}}

\newcommand{\emptyword}{\varepsilon}

\newrobustcmd{\uone}{\kl[\uone]{\ensuremath{\textnormal{U}_1}\xspace}}
\knowledge{\uone}{notion}

\newrobustcmd{\blockproduct}{\kl[\blockproduct]{\Box}}
\knowledge{\blockproduct}{notion}

\newrobustcmd{\algomegasymb}{\kl[\algomegasymb]{\ensuremath{{\boldsymbol {\tau}}}}}
\knowledge{\algomegasymb}{notion}
\newrobustcmd{\algomegastarsymb}{\kl[\algomegastarsymb]{\ensuremath{\boldsymbol \tau^*}}}
\knowledge{\algomegastarsymb}{notion}
\newrobustcmd{\algshufflesymb}{\kl[\algshufflesymb]{\ensuremath{\boldsymbol \kappa}}}
\knowledge{\algshufflesymb}{notion}

\newcommand{\shuffle}{\eta}

\newcommand{\omegaoper}[1]{\ensuremath{#1^{\algomegasymb}}}
\newcommand{\omegastaroper}[1]{\ensuremath{#1^{\algomegastarsymb}}}
\newcommand{\shuffleoper}[1]{\ensuremath{#1^{\algshufflesymb}}}

\newrobustcmd{\productoper}{\kl[\productoper]{\ensuremath{\boldsymbol \cdot}}}
\knowledge{\productoper}{notion}

\newcommand\rationals{\mathbb{Q}}

\newcommand\nats{\mathbb{N}}

\newcommand{\natorder}{\omega}
\newcommand{\negorder}{\omega^*}
\newcommand{\intorder}{\delta}
\newcommand{\rationalorder}{\shuffle}
\newcommand{\integers}{\mathbb{Z}}

\newrobustcmd{\dom}[1]{\kl[\dom]{\mathit{dom}(#1)}}
\knowledge {\dom}{notion}

\newrobustcmd{\alphabet}{\kl[\alphabet]{\ensuremath{\Sigma}}}
\knowledge {\alphabet}{notion}

\newrobustcmd{\subword}[2]{\kl[\subword]{#1|_{#2}}}
\knowledge{\subword}{notion}

\knowledge{\circplussymbol}{notion}

\newrobustcmd{\wordstar}{\kl[\wordstar] {\circlesymbol}}
\knowledge{\wordstar}{notion}
\newrobustcmd{\wordplus}{\kl[\wordplus] {\circplussymbol}}
\knowledge{\wordplus}{notion}

\newrobustcmd{\omegasymb}{\kl[\omegasymb] {\omega}}
\knowledge{\omegasymb}{notion}
\newrobustcmd{\omegastarsymb}{\kl[\omegastarsymb] {\omega^*}}
\knowledge{\omegastarsymb}{notion}
\newrobustcmd{\shufflesymb}{\kl[\shufflesymb] {\shuffle}}
\knowledge{\shufflesymb}{notion}

\newrobustcmd{\circmonoid}{\kl[\circmonoid]{\circlesymbol\text{-}monoid\xspace}}
\knowledge{\circmonoid}[\circmonoids]{notion}

\newrobustcmd{\circalgebra}{\kl[\circalgebra]{\circlesymbol\text{-}algebra\xspace}}
\knowledge{\circalgebra}[\circalgebras]{notion}

\newrobustcmd{\circalgebras}{\kl[\circalgebras]{\circlesymbol\text{-}algebras\xspace}}
\newrobustcmd{\circmonoids}{\kl[\circmonoids]{\circlesymbol\text{-}monoids\xspace}}

\newcommand\words[1]{#1^{\wordstar}}
\newcommand{\nonemptywords}[1] {#1^{\wordplus}}

\newcommand{\omegaword}[1] {\ensuremath{{#1}^{\omegasymb}}}
\newcommand{\omegastarword}[1] {\ensuremath{{#1}^{\omegastarsymb}}}
\newcommand{\shuffleword}[1] {\ensuremath{{#1}^{\shufflesymb}}}

\newcommand{\finitewords}[1]{\ensuremath{#1^*}}

\newcommand{\circlesymbol}{\ensuremath{\ostar}}

\newcommand{\ostar}{\circledast}

\newcommand{\circplussymbol}{\ensuremath{\oplus}}

\newcommand{\defs}{::=}

\newcommand{\ifthen}{\rightarrow}

\newrobustcmd{\product}{\kl[\product]{\pi}}
\knowledge{\product}{notion}

\newrobustcmd{\unit}{\kl[\unit]{\ensuremath{\mathtt{id}}\xspace}}
\knowledge{\unit}{notion}

\newcommand{\zero}{\ensuremath{0}\xspace}

\newcommand{\monoid}{\ensuremath{\mathbf{M}}\xspace}
\newcommand{\monoidset}{\ensuremath{M}}
\newcommand{\monoidN}{\ensuremath{\mathbf{N}}\xspace}
\newcommand{\monoidK}{\mathbf{K}}

\newcommand{\monoidOf}[1]{\ensuremath{(#1,\unit,\product)}\xspace}
\newcommand{\algebraOf}[1]{\ensuremath{(#1,\unit,\productoper,\algomegasymb, \algomegastarsymb,\algshufflesymb)}\xspace}

\newcommand{\monoidGap}{{\ensuremath{\mathrm{Gap}}}}

\newcommand{\gJ}{{\mathbin{\ensuremath{{\mathcal{J}}}}}}

\newcommand{\gR}{{\mathbin{\ensuremath{{\mathcal{R}}}}}}

\newcommand{\gL}{{\mathbin{\ensuremath{{\mathcal{L}}}}}}

\newcommand{\gH}{{\mathbin{\ensuremath{{\mathcal{H}}}}}}

\newcommand{\Jg}{>_{\gJ}}

\newcommand{\Jeq}{\gJ}
\newcommand{\nJeq}{\cancel{\gJ}}

\newcommand{\Req}{\gR}
\newcommand{\Leq}{\gL}
\newcommand{\Heq}{\gH}

%% file: knowledgeinfintaryquantifiers.tex
\knowledge{latex}{text=\LaTeX}

\knowledge {e.g.}{italic}
\knowledge {i.e.}{italic}
\knowledge {resp.}{italic}
\knowledge {iff}{italic}
\knowledge {a priori}{italic}
\knowledge {dom}{link=domain,text=\ensuremath{\mathit{dom}}}

\knowledge{interval}[intervals]{notion}

\knowledge{commutative}{notion}
\knowledge{shuffle-trivial}{notion}
\knowledge{shuffle-idempotent}{notion}
\knowledge{Simon's theorem}[Simon's Theorem]{notion}
\knowledge{piecewise testable}[piecewise testable languages]{notion}
\knowledge{$J$-trivial}{notion}
\knowledge{shuffle-power-trivial}{notion}
\knowledge{gap-nesting-length}{notion}

\knowledge{identity element}{notion}
\knowledge{morphism}[morphisms]{notion}
\knowledge{divides}[division]{notion}
\knowledge{derived operations}[derived operation|derived operators]{notion}
\knowledge{evaluation tree}[evaluation trees]{notion}
\knowledge{Green's relations}[relations of Green]{notion}
\knowledge{value}[values]{notion}
\knowledge{\uone}{notion}
\knowledge{block product}[block products]{notion}
\knowledge{block product principle}{notion}
\knowledge{direct product}[direct products]{notion}

\knowledge {regular element}{notion}
\knowledge {regular $\gJ$-class}[regular $\gJ$-classes|Regular $\gJ$-classes]{notion}
\knowledge {$\gL$-class}[$\gL$-classes]{notion}
\knowledge {$\gR$-class}[$\gR$-classes]{notion}

\knowledge {concatenation witness}[concatenation witnesses|Concatenation witnesses]{notion}
\knowledge {shuffle witness}[Shuffle witnesses|shuffle witnesses]{notion}
\knowledge {nice concatenation witness}{notion}
\knowledge {omega witness}[omega witnesses|Omega witnesses]{notion}
\knowledge {omega$^*$ witness}[omega$^*$ witnesses|Omega$^*$ witnesses]{notion}
\knowledge {letter witness}[Letter witnesses|letter witnesses]{notion}

\knowledge {definable cuts}[definable cut|Definable cuts]{notion}

\knowledge {letter witness}{notion}
\knowledge {omega witness}{notion}
\knowledge {omega$^*$ witness}{notion}
\knowledge {shuffle witness}[shuffle witnesses|Shuffle witnesses]{notion}

\knowledge{i->gi}{notion,text=$\mathtt{i{\rightarrow}gi}$}
\knowledge{oi->gi}{notion,text=$\mathtt{oi{\rightarrow}gi}$}
\knowledge{o*i->gi}{notion,text=$\mathtt{o^{*}i{\rightarrow}gi}$}
\knowledge{sc->sh}{notion,text=$\mathtt{sc{\rightarrow}sh}$}
\knowledge{sh->ss}{notion,text=$\mathtt{sh{\rightarrow}ss}$}
\knowledge{aperiodicity}{notion}
\knowledge{aperiodic}{link=aperiodicity,text=aperiodic}

\knowledge{Shuffle simple}[shuffle simple idempotent|shuffle simple]{notion}
\knowledge{shuffle idempotent}[shuffle idempotents|Shuffle idempotents]{notion}
\knowledge{gap insensitive idempotent}[Gap insensitive|gap insensitive]{notion}
\knowledge{ordinal idempotent}[ordinal idempotents|Ordinal idempotents]{notion}
\knowledge{ordinal* idempotent}[ordinal* idempotents|Ordinal* idempotents]{notion}
\knowledge{scattered idempotent}[Scattered idempotents|scattered idempotents]{notion}

\knowledge{$\eigi$}{notion}
\knowledge{$\scish$}{notion}
\knowledge{$\shiss$}{notion}
\knowledge{$\oigi$}{notion}
\knowledge{$\ostigi$}{notion}

\knowledge{rank}[Infinitary rank]{notion}
\knowledge {linear ordering}[orderings|linear orderings|Linear orderings]{notion}
\knowledge {subordering}[suborderings]{notion}
\knowledge {successor}[successors|Successors]{notion}
\knowledge {predecessor}[predecessors|Predecessors]{notion}
\knowledge {condensed linear ordering}{notion}
\knowledge {countable linear ordering}[Countable linear orderings|countable linear orderings]{link=linear ordering}
\knowledge {scattered linear ordering}[scattered|scattered linear orderings|Scattered linear orderings]{notion}

\knowledge {generalized sum}{notion}
\knowledge {generalized concatenation}[concatenation]{notion}
\knowledge {generalized associativity}{notion}
\knowledge {dense}{notion}
\knowledge {convex}{notion}
\knowledge {condensation}[condensations|condensed|Condensations]{notion}
\knowledge {convex subsets}{link=convex}
\knowledge {consecutive}{notion}
\knowledge {Dedekind cut}[cut|Cuts|cuts|Dedekind cuts]{notion}
\knowledge {unnatural cut}[Unnatural cuts|unnatural cuts]{notion}
\knowledge {natural cut}[natural cuts|Natural cuts]{notion}
\knowledge {gap}{notion}
\knowledge {alphabet}{notion}

\knowledge {well ordering}[well ordered|well ordered set|well ordered sets]{notion}
\knowledge {ordinal}{notion}

\knowledge {word of countable domain}[\words|countable word|countable words|words of countable domain|Words of countable domain|Countable words]{notion}
\knowledge {word}[words|Words|Countable Words]{link=countable word}
\knowledge {position}[positions|positions]{notion}
\knowledge {domain}[Domains|domains]{notion}
\knowledge {perfect shuffle}[shuffle]{notion}
\knowledge {omega power}[omega word | omega]{notion}
\knowledge {omega$^*$ power}[omega$^*$ word | omega$^*$]{notion}
\knowledge {kind}[kinds|kinds]{notion}
\knowledge {subword}[subwords]{notion}
\knowledge {factor}[factors]{notion}
\knowledge{letters}[letter]{notion}

\knowledge {empty word}{notion}
\knowledge {perfect shuffle}{notion}
\knowledge {perfectly dense}{notion}

\knowledge {regular language}[Regular languages|regular languages]{notion}
\knowledge {language}[Languages|languages]{notion}

\knowledge {o-monoid}{text=$\circ$-monoid,link=$\circ$-monoid}
\knowledge {o-monoids}{text=$\circ$-monoids,link=$\circ$-monoid}

\knowledge {o-semigroup}{text=$\circ$-semigroup,link=$\circ$-semigroup}
\knowledge {o-semigroups}{text=$\circ$-semigroups,link=$\circ$-semigroup}

\knowledge {syntactic o-monoid}{text=syntactic $\circ$-monoid,link=syntactic $\circ$-monoid}
\knowledge {syntactic o-monoids}{text=syntactic $\circ$-monoids,link=syntactic $\circ$-monoid}
\knowledge {Syntactic o-monoids}{text=Syntactic $\circ$-monoids,link=syntactic $\circ$-monoid}

\knowledge {free o-monoid}{text=free $\circ$-monoid,link=free $\circ$-monoid}
\knowledge {free o-monoids}{text=free $\circ$-monoids,link=free $\circ$-monoid}
\knowledge {Free o-monoids}{text=Free $\circ$-monoids,link=free $\circ$-monoid}

\knowledge {o-algebra}{text=$\circ$-algebra,link=$\circ$-algebra}
\knowledge {o-algebras}{text=$\circ$-algebras,link=$\circ$-algebra}

\knowledge {semigroup}[Semigroups|semigroups]{notion}
\knowledge {monoid}[monoids|Monoids]{notion}
\knowledge{product}[products|Products]{notion}
\knowledge{unit}{notion}
\knowledge{zero}{notion}
\knowledge {idempotent}[idempotents|Idempotents]{notion}

\knowledge {aperiodic monoids}[aperiodic|Aperiodic monoids]{notion}
\knowledge {recognition}[recognized|recognizes|recognize|recognizability|recognizable]{notion}
  
\knowledge {\product}{link=$\circ$-monoid}

\knowledge{\monoid}{notion}
\knowledge{\monoidPowerset}{notion}

\knowledge{identity}[equation|equational|identities|Identities]{notion}

\knowledge {$\gJ $-class}[\gJ-class|\gJ-classes|$\gJ$-classes]{notion}
\knowledge {J-class}{link = $\gJ $-class,text=$\gJ $-class}
\knowledge {J-classes}{link = $\gJ $-class,text=$\gJ $-classes}

\knowledge {\legJ}[\leq _\gJ]{notion}
\knowledge {$\gJ $-equivalent}[\gJ]{notion}
\knowledge {regular \gJ-class}[non-regular \gJ-classes|non-regular \gJ-class|regular \gJ-classes|Regular \gJ-classes]{notion}

\knowledge {ordinal regular}{notion}
\knowledge {ordinal$^*$ regular}{notion}
\knowledge {shuffle regular}{notion}
\knowledge {scattered regular}{notion}
\knowledge {shuffle simple}{notion}

\knowledge{infinitary quantifiers}[infinitary quantifier|Infinitary Quantifier]{notion}
\knowledge {language}{notion}
\knowledge {definable}[defined]{notion}
\knowledge {MSO}[monadic second-order logic|Monadic second-order logic|monadic logic|Monadic logic|\mso]{notion}
\knowledge {first-order logic}[First-order logic|FO|first-order|\fo]{notion}
\knowledge {first-order variable}[first-order variables|First-order variables]{notion}
\knowledge {monadic variable}[monadic variables|Monadic variables]{notion}
\knowledge {monadic second-order variable}[Monadic second-order variables|monadic second-order variables]{link=monadic variable}
\knowledge {weak monadic second-order logic}[weak monadic second order logic|WMSO|Weak monadic second-order logic|\wmso]{notion}
\knowledge {FO[cut]}[first-order logic on cuts|\focut]{notion}
\knowledge {MSO[finite,cut]}[weak monadic second-order on cuts|\wmsocut]{notion}
\knowledge {MSO[ordinal]|\foord}{notion}
\knowledge {MSO[scattered]|\foscat}{notion}
\knowledge {model}{notion}
\knowledge {restricted quantifier}[Restricted quantifiers|restricted quantifiers]{notion}
\knowledge {restricted set quantifier}[restricted set quantifiers|Restricted set quantifiers]{link=restricted quantifier}
\knowledge {set quantifier}[set quantifiers]{notion}
\knowledge {existential set quantifier}{link=set quantifier}
\knowledge{\rexists}{notion}
\knowledge{\rforall}{notion}

\knowledge{\Bin}{notion}
\knowledge{\Bnin}{notion}
\knowledge {\projection}{notion}
\knowledge {restricted projection}[\restrictedProjection]{notion}
\knowledge {\restrictedMonoidPowerset}{notion}

\knowledge {\Win}{notion}
\knowledge{\proj}{notion}

\knowledge{implicit operation}[Inmplicit operations|implicit operations]{notion}

\knowledgedefault{}

%% file: abstract.tex
We contribute to the refined understanding of the language-logic-algebra
interplay in the context of {\em first-order} properties of countable
words. We establish decidable algebraic characterizations of 
one variable fragment of FO as well as boolean closure of existential
fragment of FO via a strengthening of Simon's theorem about piecewise
testable languages. We propose a new extension of FO which admits infinitary
quantifiers to reason about the inherent infinitary properties of countable
words. We provide a very natural and hierarchical  block-product based 
characterization of the new extension. We also explicate its role
in view of other natural and classical logical systems such as WMSO and
FO[cut] - an extension of FO where quantification over Dedekind-cuts is 
allowed. We also rule out the possibility of a finite-basis for
a block-product based characterization of these logical systems. 
Finally, we report simple but novel algebraic characterizations of one 
variable fragments of the hierarchies of the new proposed extension of FO.

%% file: intro.tex
\newcommand{\myostar}{\circledast}
Over finite words, we have a foundational language-logic-algebra connection (see \cite{wolfgang,Pin_syntactic}) which
equates regular-expressions, \mso-logic, and (recognition by) finite monoids/automata.
In fact, one can effectively associate, to a regular language, its finite {\em syntactic monoid}.
This canonical algebraic structure carries a rich amount of information about the corresponding language.
Its role is highlighted by the classical Schutzenberger-McNaughton-Papert theorem 
(see, for instance, \cite{MPRI-notes}) which shows that {\em aperiodicity} property 
of the syntactic monoid coincides with describability using star-free expressions as well
as definability in First-Order (\fo) logic. 
So, we arrive at a refined understanding of the 
language-logic-algebra connection to an important 
subclass of regular languages: it equates star-free regular expressions, \fo-logic, 
and aperiodic finite monoids.

A variety of algebraic tools have been developed and crucially used
to obtain deeper insights. Some of these tools \cite{MPRI-notes,str_cirBook,DBLP:journals/corr/StraubingW15}
are: ordered monoids, the so-called Green's relations, wreath/block products and 
related principles etc. Let us mention Simon's celebrated theorem \cite{simon} - 
which equates piecewise-testable languages, Boolean closure of the existential fragment 
of \fo-logic and $J$-trivial finite monoids\footnote{It refers to $J$ - one of the fundamental Green's relations}.
It is important to note that this is an effective characterization, that is, 
they provide a decidable characterization of the logical fragment. 
There have been several results of this kind (see the survey 
\cite{gastin_smallfragments}).
Another particularly
interesting set of results is in the spirit of the fundamental 
Krohn-Rhodes theorem. These results establish a block-product based
decompositional characterization of a logical fragment and have 
many important applications \cite{str_cirBook}.
The prominent examples are a characterization of \fo-logic (resp. \fotwo, the two-variable fragment) in terms of strongly (resp. weakly) 
iterated block-products of copies of the unique $2$-element aperiodic monoid.  

One of the motivations for this work is to establish similar results in the theory
of regular languages of "countable words". We use 
the overarching algebraic framework developed in
the seminal work \cite{carton_MSOalgebra} to reason about languages
of "countable words". 
This framework extends the 
language-logic-algebra interplay to the setting of countable words. 
It develops fundamental algebraic structures such as {\em finite}
\circmonoids\ and \circalgebras\ and equates \mso-definability with "recognizability" by these algebraic structures.
A detailed study of a variety of sub-logics of \mso\ over "countable words" is carried out in \cite{colcombetCLO}. 
This study also extends classical Green's relations to \circalgebras\ and makes heavy use of it. 
Of particular interest to us are the results about algebraic "equational" characterizations of \fo, \focut\ -- an 
{\em extension} of \fo\ that allows quantification over "Dedekind cuts" and
\wmso\ -- an {\em extension} of \fo\ that allows quantification over finite 
sets. A decidable algebraic characterization of \fotwo\
over "countable words" is also presented in \cite{sav_fo2clo}.
Another recent development \cite{DBLP:conf/lics/AdsulSS19} is the seamless integration of
"block products" into the countable setting. The work introduces the "block product"
operation of the relevant algebraic structures and establishes an appealing "block product principle". 
Further, it naturally extends the above-mentioned "block product" characterizations of \fo\ and \fotwo\ to "countable words".

In this work, we begin our explorations into the {\em small} fragments of \fo\ over "countable words", guided by the choice of results in \cite{gastin_smallfragments}. We arrive at the language-logic-algebra connection for \foone\ -- the one variable fragment of FO. Coupled with earlier results about \fotwo\ and \fo=\fothree\
(see \cite{GHR}), 
this completes our algebraic understanding of \fo\ fragments defined by the number of permissible variables.
We next extend "Simon's theorem" on "piecewise testable languages" to "countable words" and provide a natural algebraic
characterization of the Boolean closure of the existential-fragment of \fo.
Fortunately or unfortunately, depending on the point of view, this 
landscape of small fragments of \fo\ over "countable words" 
parallels very closely the same landscape over finite words. This can be
attributed to the limited expressive power of \fo\ over "countable words".
For instance, B\`es and Carton \cite{bes_foscat} showed that the seemingly natural `finiteness' property (that the
set of all positions is a finite set) of "countable words" can not be expressed in \fo! 

One of the main contributions of this work is the introduction of 
new {\em infinitary} quantifiers to \fo. The works
\cite{dec_and_gen_qua,gradel_fotwo} also extend \fo\ over arbitrary structures by 
{\em cardinality/finitary-counting} quantifiers and studies
decidable theories thereof. An extension of \fo\ over finite and $\omega$-words by
{\em modulus-counting quantifiers} is algebraically characterized in \cite{straubing_genQuant}.
The main purpose
of our new quantifiers is to naturally allow expression of 
infinitary features which are inherent in the countable setting and study
the resulting definable formal languages in the algebraic 
framework of \cite{carton_MSOalgebra}.
An example formula using such an infinitary quantifier is: 
$\ei 1 x: a(x) \land \neg \ei 1 x: b(x)$. In its natural
semantics, this formula with one variable asserts that there are infinitely
many $a$-labelled positions and only finitely many $b$-labelled positions. 
We propose an extension of \fo\ called \foinf\ that supports {\em first-order} 
"infinitary quantifiers" of the form $\ei k x$ to talk about existence
of higher-level infinitely (more accurately, "Infinitary rank" $k$) many witnesses $x$.
We organize \foinf\ in a natural hierarchy based on the maximum allowed 
infinitary-level of the quantifiers.

We now summarize the key technical results of this paper.
We establish a hierarchical "block product" based characterization of 
\foinf. Towards this, we identify 
an appropriate simple family of \circalgebras\ and show that
this family (in fact, its initial fragments) serve as a basis 
in our hierarchical "block product" based characterization.
We establish that \foinf\ properties can be expressed simultaneously 
in \focut\ as well as \wmso. 
We also show that 
the language-logic-algebra connection for \foone\ admits novel generalizations
to the one variable fragments of the new extension of \fo. 
We finally present `no finite "block product" basis' theorems for our 
\fo\ extensions, \focut, and the class $\focut \cap \wmso$. 
This is in contrast with \cite{DBLP:conf/lics/AdsulSS19} 
where the unique 2-element \circalgebra\ is a basis for
a block-product based characterization of \fo. 

The rest of the paper is organized as follows. 
Section \ref{sec:prelims} recalls basic notions about "countable words" and summarizes the necessary algebraic
background from the framework \cite{carton_MSOalgebra}. Section \ref{sec:fo} deals with the small fragments
of \fo: \foone\ and the Boolean closure of the existential fragment of \fo. 
Section \ref{sec:fo-inf} contains the extensions \foinf\ and results relevant to it. 
Section \ref{sec:nofinbasis} is concerned with `no finite "block product" basis' theorems. 

%% file: prelims-small.tex
\newtheorem{prop}{Proposition}
\newtheorem{defn}{Definition}
\section{Preliminaries}
\label{sec:prelims}
In this section we briefly recall the algebraic framework developed in \cite{carton_MSOalgebra}.
\paragraph{{\bf Countable words}}
\label{subsection:linear orderings}
\AP
A {\em countable} \intro{linear ordering} (or simply ordering) $\alpha = (X,<)$ is a non-empty countable set $X$ 
equipped with a total order: $X$ is the \intro{domain} of $\alpha$. An ordering
$\beta = (Y,<)$ is called a \intro{subordering} of $\alpha$ if $Y \subseteq X$ and
the order on $Y$ is induced from that of $X$.
We denote by $\natorder, \negorder, \intorder, \rationalorder$ the orderings $(\nats, <), (-\nats, <), (\integers,<), (\rationals,<)$ respectively. 
A \intro{Dedekind cut} (or simply a \kl{cut}) is a left-closed subset $Y \subseteq X$ of $\alpha$.
%
Given disjoint \kl{linear orderings} $(\beta_i)_{i\in\alpha}$ indexed with a \kl{linear ordering}~$\alpha$, their
\intro{generalized sum} $\sum_{i\in\alpha}\beta_i$ is the \kl{linear ordering} over the
 union of the domains of the $\beta_i$'s, with the order defined by $x<y$ if either
$x\in\beta_i$ and $y\in\beta_j$ with $i<j$, or $x,y\in\beta_i$ for some~$i$, and $x<y$ in
$\beta_i$.  
The book \cite{rosenstein} contains a detailed study of \kl{linear orderings}.

\AP
An \intro{alphabet} $\intro{\alphabet}$ is a finite set of symbols called \intro{letters}. 
Given a \kl{linear ordering} $\alpha$, a \intro{countable word} (henceforth called \kl{word}) over $\alphabet$ of \kl{domain} $\alpha$ is a mapping 
$w: \alpha \rightarrow \alphabet$. 
The \kl{domain} of a \kl{word} is denoted $\intro[\dom]{}\dom w$. For a subset $I \subseteq \dom w$, $\subword{w}{I}$ denotes the \intro{subword} got by restricting $w$ to the domain $I$. If $I$ is an ""interval"" ($\forall x,y \in I$, $x < z < y \rightarrow z \in I$) then $\intro[\subword]{}\subword{w}{I}$ is called a \intro{factor} of $w$.
The set of all words is denoted $\intro[\wordstar]{}\words \alphabet$ and the set of all non-empty (resp. finite) words $\intro[\wordplus]{}\nonemptywords \alphabet$ (resp. $\finitewords \alphabet$). A \intro{language} (of countable words) is a subset of $\words\alphabet$. 
%
%
%
The \intro{generalized concatenation} of the \kl{words} $(w_i)_{i\in\alpha}$ indexed by a \kl{linear ordering}~$\alpha$ is $\prod_{i \in \alpha} w_i$ and denotes the \kl{word} $w$ of \kl{domain} $\sum_{i \in \alpha} \beta_i$ where $\beta_i$ are disjoint and such that $\subword{w}{\beta_i}$ is isomorphic to $w_i$ for all $i \in \alpha$.

\AP
The \intro{empty word} $\emptyword$, is the only \kl{word} of empty \kl{domain}. The \intro{omega power} of a word $w$ is defined as $\intro[\omegasymb]{}\omegaword w \defs \prod_{i \in \omega} w$. 
The \intro{omega$^*$ power} of a word $w$, denoted by
$\intro[\omegastarsymb]{}\omegastarword{w}$, is the concatenation of omega$^*$ many $w$'s.
The \intro[shuffle]{perfect shuffle} for a non-empty finite set of \kl{letters} $A
\subseteq \alphabet$ (denoted by {$\shuffleword A$}) is a \kl{word} of
\kl{domain}~$(\rationals,<)$ in which only letters from $A$ occur and, all non-empty and non-singleton "intervals" contain at least one
occurrence of each letter in $A$. This \kl{word} is unique up to isomorphism
\cite{shelah_MSO}. We can extend the notion of "perfect shuffle" to a finite set of words
$W = \{w_1,\dots,w_k\}$. We define $\intro[\shufflesymb]{}\shuffleword W$ to be $\prod_{i
\in \rationals} w_{f(i)}$ where $f: (\rationals, <) \to \{1,2,\dots,k\}$ is the unique "perfect shuffle" over the set of letters $\{1,2,\dots,k\}$.

\paragraph{{\bf The algebra}}
\label{subsection:o-monoids}
\AP
A  \intro[\circmonoid]{\circmonoid} \intro[\product]{}\intro[\unit]{}$\monoid = \monoidOf \monoidset$ is a set $\monoidset$ equipped with an operation $\product$, called the \intro{product}, from $\words \monoidset$ to $\monoidset$, that satisfies $\product(a)=a$ for all $a\in \monoidset$, and the \intro{generalized associativity} property: for every \kl{words} $u_i$ over $\monoidset$ with $i$ ranging over a \kl{countable linear ordering} $\alpha$,
$\product\left(\prod_{i\in\alpha}u_i\right)=\product\left(\prod_{i\in\alpha}\product(u_i)\right)$.
We reserve the notation $\unit$ for the \intro{identity element}
$\unit=\product(\emptyword)$;
it is called the neutral element in \cite{carton_MSOalgebra}. An example of a
\circmonoid\ is the free \circmonoid\ $(\words \alphabet, \emptyword, \prod)$ over the
alphabet $\alphabet$ with the product being the "generalized concatenation".
Now we 
discuss some natural algebraic notions. 
A ""morphism"" from a \circmonoid\ $(\monoidset,\unit,\pi)$ to a \circmonoid\ $(\monoidset',\unit',\pi')$ is a map $h: \monoidset \rightarrow \monoidset'$ such that, for
every $w \in \words{\monoidset}$, $h(\pi(w))=\pi'(\bar{h}(w))$ where $\bar{h}$ is the pointwise extension of $h$ to words.
We skip the notions sub-$\circmonoid$ and direct products since they are as expected. 
We say $\monoid = (\monoidset,\unit, \pi)$ \intro{divides} $\monoid' = (\monoidset',\unit', \pi')$ if there exists a sub "\circmonoid"
$\monoid'' = (\monoidset'', \unit'', \pi'')$ of $\monoid'$ and a surjective morphism from $\monoid''$ to $\monoid$.

%
%
\AP
A \circmonoid\ $\monoid = (\monoidset,\unit, \pi)$ 
is said to be finite if \monoidset\ is so. Note that, even for a finite \kl{\circmonoid}, 
the product operation $\product$ has an infinitary description. It turns out that $\product$ can be captured using finitely presentable \intro{derived operations}.\intro[\algomegasymb]{}\intro[\algomegastarsymb]{}\intro[\algshufflesymb]{}\intro[\productoper]{}
Corresponding to a "\circmonoid" $\monoidOf{\monoidset}$ there is an induced \intro[\circalgebras]{\circalgebra} $\monoid = \algebraOf \monoidset$ where the operations are defined as following: for all $a, b \in \monoidset$, $a \productoper b = \product(ab)$, $\omegaoper a = \product(\omegaword a)$, $\omegastaroper a = \product(\omegastarword a)$ and for all $\emptyset \neq E \subseteq \monoidset$, $\shuffleoper E = \product(\shuffleword E)$. 
These "derived operators" satisfy certain natural axioms; see \cite{carton_MSOalgebra} for
more details. 
It has been established in \cite{carton_MSOalgebra} that 
an {\em arbitrary} finite \circalgebra\ $\monoid = \algebraOf \monoidset$ satisfying these 
natural axioms is induced by a unique \circmonoid\ $\monoid = \monoidOf \monoidset$.
We later introduce the notion of an \emph{evaluation tree} which aids this correspondence.
It is rather straightforward to define the notions of morphisms, subalgebras,
direct-products as well as division for \circalgebras.

\AP
It follows from the definition of a \circalgebra\ $\monoid = \algebraOf \monoidset$ that $(\monoidset, \unit, \productoper)$ is a ""monoid"",
that is the operation $\productoper$ is associative with identity $\unit$. 
For a singleton set $E=\{m\}$, we write $\shuffleoper m =\shuffleoper{\{m\}}$. 
Notice that for all $m \in \monoidset$, $m \productoper \unit = \unit \productoper m = m$ and
for all $\emptyset \neq E \subseteq \monoidset$, $\shuffleoper E = \shuffleoper{(E \cup \{
\unit \})}$.
Further, $\omegaoper{\unit}=\omegastaroper{\unit}=\shuffleoper \unit=\unit$.
As a result, in our definitions of \circalgebras\  later in the paper, we restrict
the descriptions of "derived operators" to $\monoidset \setminus \{\unit\}$.

\AP
An ""evaluation tree"" over a word $u \in \words{\monoidset} \backslash \{\emptyword\}$ is a tree ${\mathcal T} = (T, h)$
such that every branch/path of ${\mathcal T}$ is of finite length and 
where every vertex in $T$ is a factor of $u$, the root is $u$ and $h:T \rightarrow \monoidset$ is a map such that:
\begin{itemize}
\item A leaf is a singleton letter $a \in \monoidset$ such that $h(a) = a$.
\item Internal nodes have either two or $\omega$ or $\omega^*$ or $\rationals$ many children.
\item If $w$ has children $v_1$ and $v_2$, then $w=v_1v_2$ and $h(w) = h(v_1) \productoper h(v_2)$.
\item If $w$ has $\omega$ many children $\langle v_1, v_2, \dots \rangle$, then there is an 
	""idempotent""\footnote{An "idempotent" is an element $e$ where $e\productoper e =
	e$} $e$ such that $e = h(v_i)$ for all $i \geq 1$, and $w = \prod_{i \in \omega} v_i$ and $h(w) = \omegaoper{e}$.
\item If $w$ has $\omega^*$ many children $\langle \dots, v_{-2}, v_{-1} \rangle$, then there is an "idempotent" $f$ such that $f = h(v_i)$ for all $i \leq -1$, and $w = \prod_{i \in \omega^*} v_i$ and $h(w) = \omegastaroper{f}$.
\item If $w$ has $\rationals$ many children $\langle v_i \rangle_{i \in \rationals}$, then $w = \prod_{i \in \rationals} v_{i}$ where for the "perfect shuffle" $f$ over an $E=\{a_1,\dots, a_k\} \subseteq \monoidset$, $h(v_{i}) = a_{f(i)}$, and $h(w) = \shuffleoper E$.
\end{itemize}
\AP
The ""value"" of ${\mathcal T}$ is defined to be $h(u)$. It was shown in \cite[Proposition
8 and 9]{carton_MSOalgebra} that every word $u$ has an "evaluation tree" and the "values"
of two "evaluation trees" of $u$ are equal and they are equal to $\product(u)$. Therefore,
a \circalgebra\ defines the "generalized associativity" product $\product:\words{\monoidset}
\rightarrow \monoidset$. The correspondence between finite \circmonoids\ and
\circalgebras\ permits interchangeability; we exploit it implicitly.

\AP
A morphism from the {\em free} $\circmonoid\ \words{\alphabet}$ to
$\monoid$ is described (determined) by a map $h': \alphabet \to \monoidset$; we simply
write $h': \alphabet \to \monoid$. With 
$h'$ also denoting its pointwise extension 
$h': \words{\alphabet} \to \words{\monoidset}$, given a word $u \in \words{\alphabet}$, 
we can use the evalution tree over the word $h'(u) \in \words{\monoidset}$ to
obtain $\product(h'(u)) \in \monoidset$. By further abuse of notation, $h':\words{\alphabet}
\to \monoidset$ also denotes the morphism which sends $u$ to $\product(h'(u))$. We say that $L$ is recognized by $\monoid$ if there exists a map/morphism $h': \words \alphabet \to \words \monoid$ such that $L=h'^{-1} (h'(L))$.
The fundamental result of \cite{carton_MSOalgebra} states that ""regular languages""
(\mso\ definable "languages") are exactly those "recognized" by finite \circmonoids\
(equivalently \circalgebras). It is important to note that, every regular language $L$
is associated a finite (canonical/minimal) syntactic \circmonoid\ which divides every
\circmonoid\ that recognizes $L$. Further, it can be 
represented as a \circalgebra\
from a finite description of $L$.
\AP
\begin{example}\label{example:uone}
	 The \circmonoid\ $\intro{\uone} = \monoidOf{\{\unit, \zero \}}$ and its induced \circalgebra\ are shown on the left and right respectively.
\begin{equation*}
\pi(u) = 
\begin{cases}
\unit & \text{if}\ u \in \words{\{\unit \}}\\
\zero & \text{otherwise}
\end{cases}
~~~~~~~~~~~~\begin{tabular}{l|ll|l|l}
			& \unit & \zero & $\algomegasymb$ & $\algomegastarsymb$  \\
			\hline
			\unit & \unit & \zero & \unit                     & \unit    \\
			\zero & \zero & \zero & \zero                     & \zero                                                      
\end{tabular}
~~~~~ \shuffleoper S = \begin{cases}
\unit & \text{if}\ S = \{\unit\} \\
\zero & \text{otherwise}
\end{cases}
\end{equation*}
Let $\Sigma=\{a,b\}$ and $L$ be the set of words which contain an occurence of letter $a$.
It is easy to see that the map $h: \Sigma \to \uone$ sending $h(a)=\zero, h(b)=\unit$
recognizes $L$ as $L=h^{-1}(\zero)$. In fact, $\uone$ is the syntactic \circmonoid\ of
$L$.
\end{example}
\newcommand{\cci}{\mathsmaller{[~]}}
\newcommand{\oci}{\mathsmaller{(~]}}
\newcommand{\coi}{\mathsmaller{[~)}}
\newcommand{\ooi}{\mathsmaller{(~)}}
\begin{example} \label{example:gap}
Consider the \circalgebra\ $\monoidGap$= $\algebraOf{\{\unit,\cci, \oci, \coi, \ooi, g\}}$.
We let $\Sigma=\{a\}$ and define the map $h: \Sigma \to \monoidGap$ as $h(a)= \cci$.
The resulting morphism maps a word $u$ to $h(u)=g$ iff the word $u$ admits a "gap"; that
is a "cut" with no maximum and its complement has no minimum. Other words are mapped to
their right `ends-type': for instance, $h(u)=\coi$ iff $\dom u$ has a minimum and no maximum.
For a "word" $v = \omegaword a \omegastarword a$, the pointwise extension
$v'=h(v)=\omegaword{\cci} \omegastarword{\cci}$. 
An example "evaluation tree" $\mathcal{T}$ for $v'$ consists of root with two children.
The left (resp. right) child has $\omega$ (resp. $\omega^*$) many children
$\cci$ and has "value" $\omegaoper{\cci}$ (resp. $\omegastaroper{\cci}$). 
As a result, the value of $\mathcal{T}$ is $\omegaoper{\cci} \productoper
\omegastaroper{\cci} = \coi \productoper \oci=g$.
\begin{align*}
\begin{array}{c|cccccc|c|c}
\productoper&\,\,\cci\,&\,\coi\,&\,\oci\,&\,\ooi\,&g &\,\quad&\, \algomegasymb\, &\,\algomegastarsymb\\
\hline
\cci	\,&\cci	&\coi&\cci&\coi\,	&g &\,	&\,\coi\,&\,\oci	\\
\coi	\,&\cci&\coi	& g & g\,			&g &\,	&\,\coi\,&\,\ooi		\\
\oci	\,&\oci &\ooi &\oci &\ooi \,	&g &\,	&\,\ooi \,&\,\oci	\\
\ooi	\,& \oci &\ooi &g&g\,					&g &\,	&\,g\,&\,g \\
g & g & g & g & g & g &  & g & g
\end{array}
&&
\shuffleoper S&=\begin{cases}
				\unit&\text{if }S = \{\unit\} \\
				g&\text{otherwise}
				\end{cases}
\end{align*}
\end{example}

\AP
We can characterize sets of \circmonoids\ using \intro{identities}. For example, we say that $\uone$ satisfies ""commutative"" equation $x \productoper y = y \productoper x$. This means that the equation holds for any assignment of elements in the \circmonoids\ to the variables $x$ and $y$. 
Like in the case of "monoids", the set of \circmonoids\ satisfying a set of equations are closed under subsemigroup, "division" and ""direct product"" \cite{colcombetCLO}.

\AP
The ""block product"" of \circmonoids\ \monoid and \monoidN, is denoted by $\intro[\blockproduct]{} \monoid \blockproduct \monoidN$ and is the semidirect product of $\monoid$  and $\monoidK = \monoidN^{M \times M}$ with respect to the {\em canonical} left and right `action' of $\monoid$ on $\monoidK$. The details are given in \cite{DBLP:conf/lics/AdsulSS19}. The ""block product principle"" characterizes languages defined by "block product" of \circmonoids. 
Towards this, fix a map $h: \alphabet \to \monoid \blockproduct \monoidN$ such that
$h(a)=(m_a, f_a)$ where $m_a \in \monoidset$ and $f_a: M \times M \to N$. The map
$h_1: \alphabet \to \monoid$ setting $h_1(a)=m_a$ defines a morphism $h_1: \words{\alphabet} \to
M$. We define the \emph{transducer} $\sigma: \words{\Sigma} \to \words{(M \times \Sigma
\times M)}$ as follows: let $u \in \words{\Sigma}$ with domain $\alpha$. The word
$u'=\sigma(u)$ has domain $\alpha$ and for a position $x \in \alpha$, 
$u'(x) = (h_1 (u_{< x}), u(x), h_1 (u_{> x}))$. Here $u_{< x}$ (resp. $u_{> x})$)
is the subword of $u$ on positions strictly less (resp. greater) than $x$.

\begin{proposition} [Block Product Principle \cite{DBLP:conf/lics/AdsulSS19}]
\label{prop:bpp}
\label{cor:bpp}
Let $L \subseteq \words{\alphabet}$ be recognized by 
$h: \alphabet \to \monoid \blockproduct \monoidN$ 
Then $L$ is a boolean combination of languages of the form $L_1$ and $\sigma^{-1}(L_2)$
where $L_1$ and $L_2$ are recognized by $\monoid$ and $\monoidN$ respectively and
$\sigma: \words{\alphabet} \rightarrow \words{(\monoidset \times \alphabet \times \monoidset)}$ is a state-based transducer.
\end{proposition}

%% file: fo.tex
%
%
\AP
In this section, we focus on two particularly small fragments of first-order logic
interpreted over
countable words. First-order logic uses variables $x,y,z,\ldots$ which are interpreted
as positions in the domain of a word.
The syntax of \intro{first-order logic} (\fo) is: $x<y \mid a(x) \mid \phi_1 \wedge \phi_2 \mid \phi_1 \vee \phi_2 \mid  \neg \phi \mid
\exists x ~\phi$, for all $a \in \alphabet$.

We skip the natural semantics.
 A "language" $L$ of "countable words" is said to be \fo-definable if there exists
an \fo-sentence $\phi$ such $L = \{ u \in \words{\Sigma} \mid u \models \phi\}$.

\AP
Recall that the classical Schutzenberger-McNaughton-Papert theorem characterizes
\fo-definabilty of a regular language of finite words in terms of aperiodicity of its
finite syntactic monoid.
The survey \cite{gastin_smallfragments}
presents similar decidable characterizations of several interesting small fragments of
\fo-logic such as \foone, \fotwo, $\intro{\bc}$ -- boolean closure 
of the existential first-order logic. 
It is known {\cite{GHR} that, over finite {\em as well as countable words}, \fo\ = \fothree.
As mentioned in the introduction, over countable words, we already have decidable
algebraic characterizations of \fothree\ from \cite{colcombetCLO} and \fotwo\ from
\cite{sav_fo2clo}.
%
Here we identify decidable algebraic characterizations, over "countable words", for \foone\
and $\bc$.
\subsection{FO with single variable}
\label{subsec:fo-one}
The fragment \foone\ has access to only one variable.
We recall that over finite words a "regular language" is \foone-definable iff its
syntactic monoid is "commutative" and "idempotent". We henceforth focus our attention
to \foone\ on "countable words".

Clearly, \foone\ can recognize all words with a particular "letter". With a single variable
the logic cannot talk about order of "letters" or count the number of occurrence of a "letter". 
This gives an intuition that the syntactic
\circmonoid\ of a language definable in \foone\ is "commutative" and "idempotent". 

\AP
We say that a \circalgebra\ $\monoid=\algebraOf \monoidset$ is ""shuffle-trivial"" if it satisfies the equational "identity": 
\[\shuffleoper{\{x_1, \ldots, x_p\}} = x_1 \productoper x_2 \productoper \ldots \productoper x_p\]

Then $\monoid$ is commutative: 
$x \productoper y = \shuffleoper{\{x,y\}} = \shuffleoper{\{y,x\}} = y \productoper x$. Moreover, every element of $\monoid$ is a
""shuffle-idempotent"": for all $m \in \monoidset, \shuffleoper{m}=m$. It is a consequence of the axioms of a \circalgebra\ that a "shuffle-idempotent" is
an "idempotent".


%
%


\newcommand{\letters}{\mathrm{alpha}}

\begin{theorem}\label{thm:foone}
	Let $L \subseteq \words{\al}$ be a "regular language". The following are equivalent.
\begin{enumerate}
	\item $L$ is recognized by some finite "shuffle-trivial" \circalgebra.
	\item $L$ is a boolean combination of "languages" of the form $\words B$ where $B
		\subseteq \al$.
	\item $L$ is definable in \foone.
	\item $L$ is "recognized" by direct product of $\uone$s.
	\item The syntactic \circalgebra\ of $L$ is "shuffle-trivial".
\end{enumerate}
\end{theorem}	
\begin{proof}\hfill

	\noindent ($1 \Rightarrow 2$) Let $L$ be "recognized" by $h \colon \words{\al} \to \monoid$
	where $\monoid = \algebraOf \monoidset$ is a "shuffle-trivial" \circalgebra. Note that, by \circalgebra\ axioms, for any $m \in
	M$, we have $\shuffleoper m \productoper \shuffleoper m = \shuffleoper m$. If $m$
	is a "shuffle-idempotent", that is if $\shuffleoper m = m$, we get that $m = m
	\productoper m$. Thus a "shuffle-idempotent" is necessarily an "idempotent". So $\monoid$
	is a "commutative", "shuffle-idempotent" (and hence "idempotent") \circalgebra, meaning
	its every element is a "shuffle-idempotent" (and hence an "idempotent").

	Consider an arbitrary "word" $u \in \words{\al}$ and
	let $\letters(u) \subseteq \al$ be the set of "letters" in the word $u$.  We show
	that $h(u) = h(a_1) \productoper \ldots \productoper h(a_n)$ where $\{a_1, \ldots,
	a_n\} = \letters(u)$. For the empty word, it is trivially true.  For any non-empty
	word, the proof uses "evaluation trees" introduced in Section~\ref{sec:prelims}.
	Let $\mathcal{T} = (T,h)$ be an "evaluation tree" over 
	$u$. We show by induction on the tree that $h(v) = h(a_1) \productoper \ldots 
	\productoper h(a_n)$ where $\{a_1, \ldots, a_n\} = \letters(v)$ 
	for all nodes $v$ in the tree. Consider a node $v$ of the tree.
	\begin{enumerate}
	\item Case $v$ is a letter: The induction hypothesis clearly holds.
	\item Case $v$ is a concatenation of words $v_1$ and $v_2$: This is same as in the
		classical finite word case. The induction hypothesis holds since it holds
		for both $v_1$ and $v_2$, and since all elements of $\monoid$ are
		"idempotents", and "commutative".
	\item Case $v$ is an omega sequence of words $\langle v_1, v_2, \dots \rangle$
		such that there exists an $e \in M$ and $h(v_i) = e$ for all $i \geq 1$
		and $h(v) = \omegaoper e$. Since from the axioms of \circalgebra\ 
		$\shuffleoper e = \omegaoper {(\shuffleoper e)}$ we have $h(v) = e$.
		Clearly there is a $k \geq 1$ such that $\letters(v_1v_2\dots v_k) =
		\letters(v)$ and therefore it suffices to show that $h(v_1v_2 \dots v_k) =
		h(v)$. This is true, since $e$ is an idempotent $h(v_1v_2\dots v_k) =
		h(v_1) \productoper h(v_2) \productoper \dots \productoper h(v_k) = e$.  
	 \item Case $v$ is an omega$^*$ sequence of words: This is symmetric to the case
		 above. The induction hypothesis follows from the following axiom of
		 \circmonoid: $\shuffleoper e = \omegastaroper {(\shuffleoper e)}$
	 \item Case $v$ is a "perfect shuffle" such that $h(v) = \shuffleoper{
		 \{b_1,\dots,b_k\}}$. By the shuffle-trivial property, we have $h(v) = b_1
		 \productoper \dots \productoper b_k$. Let $v = \prod_{i \in \rationals}
		 v_i$ where $h(v_i) \in \{b_1,\dots,b_k\}$. By induction hypothesis
		 $h(v_i) = h(a_1^i) \productoper \ldots \productoper h(a_n^i)$ where
		 $\letters(v_i) = \{a_1^i, \ldots, a_n^i\}$. Let $l \geq k$ and $j_1,j_2,
		 \dots,j_l \in \rationals$ be such that we get the following: $\{h(v_{j_1}),
		 h(v_{j_2}),\dots,h(v_{j_l})\} = \{b_1,\dots,b_k\}$ and
		 $\letters(v_{j_1}\dots v_{j_l}) = \letters(v)$.  Let $w = v_{j_1}\dots
		 v_{j_l}$. Since elements of $\monoid$ are "commutative" and "idempotents", $h(v)
		 = h(w) = h(v_{j_1}) \productoper \dots \productoper h(v_{j_l})$.  This
		 shows that the induction hypothesis also holds in this case, as it
		 reduces to the finite concatenation case.
	\end{enumerate}
%
%
	The induction hypothesis, therefore, holds for any word $u \in \words A$.
	So $L$ is union of equivalence classes defined by the
	finite index relation $\{(u,v) \mid \letters(u) = \letters(v)\}$. All these
	classes are boolean combination of "languages" of the form $\words{B}$ for some $B
	\subseteq \al$, as seen below.
	\[ \{u \mid \letters(u)=B \} = \words B \setminus \left ( \bigcup_{b \in B}
	\words{(B \setminus \{b\})} \right ) \]
	($2 \Rightarrow 3$) Note that $\words{B}$ is expressed by the $\foone$ formula
	$\forall x \lor_{a \in B} a(x)$. The claim follows from boolean closure of
	$\foone$.

	\noindent ($3 \Rightarrow 4$) Due to the restriction of a single variable, any formula
	$\phi(x)$ is a boolean combination of atomic "letter" predicates.  Since a
	position in a word can have exactly one "letter", any non-trivial formula $\phi(x)$ is a
	disjunction of "letter" predicates, e.g. $a(x) \lor b(x)$. A "language" defined by the
	sentence $\exists x ~(a(x) \lor
	b(x))$ is recognized by the \circmonoid\ \uone\ via $h \colon \al \to \uone$ that
	maps $a,b$ to $0 \in \uone$ and every other letter to $\unit \in \uone$. A
	"language" defined by boolean combination of such sentences can be "recognized" by
	direct products of \uone.

	\noindent ($4 \Rightarrow 5$)   The syntactic \circmonoid\ of
	$L$ "divides" any \circmonoid\ that "recognizes" $L$; so it "divides" a direct product
	of finitely many \uone.  It is almost trivially verified that \circmonoid\ 
	\uone\ is "commutative" and "shuffle-trivial".  Since these properties are
	equivalent to a set of "identities" and "identities" are preserved under direct
	product and "division", we get
	that the syntactic \circmonoid\ of $L$ is "commutative" and "shuffle-trivial".

	\noindent ($5 \Rightarrow 1$) The syntactic \circalgebra\ of $L$ is finite because
	$L$ is a "regular language". Also, it is "commutative" and "shuffle-trivial" by
	assumption, and a "language" is always recognized by its syntactic \circalgebra. 
	So this direction trivially holds.
\qed \end{proof}

We point out an interesting connection of the above result to the "block product" based
characterizations from \cite{DBLP:conf/lics/AdsulSS19}. As shown there, \fothree
(resp. \fotwo)
are characterized by strongly (resp. weakly)
iterated "block products" of copies of $\uone$. In the same spirit, $\foone$ is
characterized by direct-products of copies of $\uone$.

\subsection{Boolean closure of existential \fo}
\label{subsec:simon}
\AP
Let us first recall the characterization of $\bc$ - the boolean closure of
existential \fo\ over finite words. This is precisely the content of the theorem due
to Simon \cite{simon}. The usual presentation of "Simon's theorem" refers to
""piecewise testable languages"" which are easily seen to be equivalent to
$\bc$-definable languages.
""Simon's theorem"" states that a "regular language" of finite words is
$\bc$-definable iff its syntactic monoid is "$J$-trivial". 
We recall that a monoid $M$ is ""$J$-trivial"" if and only if for all $m, n \in M$, $MmM = MnM$ implies $m=n$. In short, the Green's equivalence relation $J$ on $M$ is the equality relation.
We refer to \cite{MPRI-notes} for a detailed study of Green's relations and its use
in the proof of "Simon's theorem". 

The original proof of "Simon's theorem" uses the congruence $\sim_n$, parametrized
by $n \in \nat$, on finite words $\finitewords \alphabet$: for $u, v \in \Sigma^*$,
$u \sim_n v$ if $u$ and $v$ have the same set of "subwords" of length less than or equal to $n$. 
Note that $\sim_n$ has finite index. 
It turns out that the finite quotient monoid
${\finitewords \alphabet}/\sim_n$ is "$J$-trivial". Furthermore, every finite "$J$-trivial" monoid 
is a quotient of the ${\finitewords \alphabet}/\sim_m$ for an appropriate choice of $\alphabet$ and $m \in
\nat$. It is known \cite{DBLP:journals/corr/StraubingW15} that a finite monoid 
$(M, \productoper)$ is "$J$-trivial" iff it satisfies the (profinite) "identities"\footnote{We denote the unique idempotent power of $m$ by $m^!$}: $x^! = x \productoper x^!  \mbox{~and~} (x \productoper y)^! = (y \productoper x)^!$.
See the proof of the following theorem in \cite[Theorem 3.13]{MPRI-notes}.
\begin{theorem}[Simon's theorem ~\cite{MPRI-notes}]
Let $\monoid = (M, \unit, \productoper)$ be a "$J$-trivial" monoid. Consider the morphism $h: \finitewords \alphabet \mapsto \monoid$. Then there exists an $n$ such that for all $u, v \in \finitewords \alphabet$, $u \sim_n v$ implies $h(u) = h(v)$.
\end{theorem}

We fix $n \in \nat$ and work with $\sim_n$ defined on "countable words"
$\words{\Sigma}$: for $u, v \in \words{\Sigma}$,
$u \sim_n v$ if $u$ and $v$ have the same set of "subwords" of length less than or equal to $n$. 
It is immediate that $\sim_n$ is an equivalence
relation on $\words{\al}$ of finite index. We let $S_n = \words{\Sigma}/\sim_n$ denote
the finite set of $\sim_n$-equivalence classes. For a word $w$,  $[w]_n$ denotes the 
$\sim_n$-equivalence class which contains $w$. 
\begin{lemma} \label{lem:sn-algebra}
There is a natural well-defined product operation 
$\product:  \words{S_n} \to S_n$ as follows:
$\product \Big(\prod_{i \in \alpha} [w_i]_n \Big) = \left [ \prod_{i \in \alpha} w_i \right ]_n$.
This operation $\product$ satisfies the "generalized associativity" property. As a result,
$\mathbf{S_n} = (S_n, \unit=[\emptyword]_n, \pi)$ is a \circmonoid. 
\end{lemma}
\begin{proof}
Let $w = \prod_{i \in \alpha} w_i$ and $w' = \prod_{i \in \alpha} w'_i$ where $w_i \sim_n
w'_i$ for all $i \in \alpha$. To prove $\pi$ is well-defined, we show that 
$\pi(w) = \pi(w')$. It suffices to show that all subwords of $w$ of length less than or
equal to $n \in \nat$ are also subwords of $w'$. 

Consider an arbitrary subword $u$ of $w$ that is of length less than or equal to $n$.  If
$u$ is the empty subword, it also is a subword of $w'$.  Otherwise, we can break $u$ into
non-empty factors $u=u_1u_2\dots u_k$  where $k \leq n$ and $u_j$ (for $1 \leq j \leq k$)
is a subword of $w_{i_j}$ ($i_j < i_{j'}$ whenever $j < j'$). Since $w_{i_j} \sim_n
w'_{i_{j}}$ and $|u_j| \leq |u| \leq n$, we have $u$ is a subword of $w'$ as well.
Therefore, $\pi$ is well-defined.

Now we turn to the issue of proving "generalized associativity" property of $\pi$. Let $u =
\prod_{i \in \alpha} u_i$ where $u_i = \prod_{j \in \alpha_i} [v_j]_n$ and $\alpha$ is any
countable "linear ordering". We have $\pi(u_i) = [\prod_{j \in \alpha_i} v_j]_n$ and hence 
\[ \pi(\prod_{i \in \alpha} \pi(u_i)) 
	= \left [ \prod_{i \in \alpha} (\prod_{j \in \alpha_i} v_j) \right ]_n =\pi(u)
\] This completes the proof. 
\qed \end{proof}

Note that the lemma implies that $h_n: \words{\Sigma} \to \mathbf{S_n}$ mapping $w$ to
$[w]_n$ is a "morphism" of $\circmonoids$.

\begin{lemma}\label{lem:fin}
Every "countable word" $u$ has a finite "subword" $\widehat{u}$ such that $u \sim_n \widehat{u}$.
\end{lemma}
\begin{proof}
Let $u$ be an arbitrary "countable word" and let $W$ be the set of "subwords" of length $n$ or less in $u$. For a $v \in W$ identify one set $X_v \subseteq \dom u$ such that $v = \subword{u}{X_v}$. Let $X = \cup_{v \in W} X_v$ and $\widehat u = \subword{u}{X}$ the "subword" corresponding to $u$ restricted to the finite set of positions in $X$. Then, $u \sim_n \widehat u$.
%
\qed \end{proof}

\AP
We say that a \circalgebra\ is ""shuffle-power-trivial"" if it satisfies 
the (profinite) "identity": 
\[\shuffleoper{\{x_1, \ldots, x_p\}} = (x_1 \productoper x_2 \productoper \ldots
\productoper x_p)^!\]

Note that, every "idempotent" of such a \circalgebra\ is a "shuffle-idempotent":
$x^! = x$ implies $\shuffleoper{x}=x$. 
Moreover it satisfies the "identities" for "$J$-trivial": 
$(x \productoper y)^! = \shuffleoper{\{x,y\}} = \shuffleoper{\{y,x\}} = (y \productoper x)^!$ and $x^! = x \productoper x^!$ follows 
from the axioms of \circalgebra\ and $x^! = \shuffleoper x = \shuffleoper{(\shuffleoper x)} = \omegaoper{(\shuffleoper x)} = \omegaoper{(x^!)} = \omegaoper{x} = x \productoper \omegaoper x$. It follows that "shuffle-power-trivial" \circalgebras\ are aperiodic.
%
%
%
\begin{theorem} \label{thm:simon-clo}
	Let $L \subseteq \words{\al}$ be a "regular language". The following are equivalent.
\begin{enumerate}
	\item $L$ is "recognized" by a finite "shuffle-power-trivial" \circalgebra.
	\item $L$ is "recognized" by the quotient "morphism" $h_n:\words{\Sigma} \to  \mathbf{S_n}$ for some $n$. 
	\item $L$ is definable in $\bc$.
	\item The syntactic \circalgebra\ of $L$ is "shuffle-power-trivial".
\end{enumerate}
\end{theorem}
\begin{proof}\hfill

\noindent $(1 \Rightarrow 2)$ Let $L$ be recognized by $h \colon \words{\al} \to \monoid$
where $\monoid = \algebraOf \monoidset$ is a "shuffle-power-trivial"
\circalgebra.  Since the "identities" are preserved in
sub-\circalgebra, we can assume $h$ to be surjective.  Consider the restriction of $h$ to
the free monoid $\finitewords \al$ resulting in the induced monoid morphism, also denoted
$h$ by slight abuse of notation, $h \colon \finitewords \al \to (M, \unit, \productoper)$.  
By the "identities" of the $\circalgebra$ $\monoid$, this morphism is surjective.

By Simon's theorem, there exists $n \in \nat$, such that
$(M , \unit, \productoper)$ is a quotient of $\fwq\al n$
and $u \sim_n v$ implies $h(u) =
h(v)$ thereby creating the quotienting morphism defined by mapping $[u]_n \mapsto h(u)$.
This implies in our morphism $h \colon \words\al \to \monoid$, for all finite words $u,v$,
we have $u \sim_n v$ implies $h(u) = h(v)$. Note that by Lemma~\ref{lem:fin}, for every countable
word $w$ there exists a finite subword of it, $\hat{w}$ such that $w \sim_n \hat{w}$. We
now show that $h(w) = h(\hat{w})$. In the remainder of this proof, $\hat{u}$ for a
countable word $u$ denotes a subword of $u$ such that $u \sim_n \hat{u}$ where 
existence of $\hat{u}$ is ensured by Lemma~\ref{lem:fin}. 

If $w$ is the empty word, then $\hat{w}$ is also the empty word, and the property holds.
Otherwise, let $\mathcal{T} = (T,h)$ be an "evaluation tree" over $w$. We prove by
induction on the tree that for every node $v$ of the tree, $h(v) = h(\hat{v})$.
\begin{enumerate}
	\item Case $v$ is a letter: The induction hypothesis clearly holds by taking
		$\hat{v} = v$.
	\item Case $v$ is a concatenation of words $v_1$ and $v_2$: Note that
		$\hat{v_1} \sim_n v_1$ and $\hat{v_2} \sim_n v_2$ implies $\hat{v_1}
		\hat{v_2} \sim_n v_1v_2$ and also $\hat{v_1}\hat{v_2}$ is a finite subword
		of $v_1v_2$.  By induction hypothesis, $h(v_1) = h(\hat{v_1})$ and $h(v_2)
		= h(\hat{v_2})$. Hence $h(v_1v_2) = h(\hat{v_1}\hat{v_2})$. This proves
		the induction hypothesis holds in this case.

	\item Case $v$ is an omega sequence of words $\langle v_1, v_2, \dots \rangle$
		such that there exists an idempotent $e \in M$ and $h(v_i) = e$ for all 
		$i \geq 1$
		and $h(v) = \omegaoper e$. Since from the axioms of \circalgebra\ 
		$\shuffleoper e = \omegaoper {(\shuffleoper e)}$ we have $h(v) = e$.
		Because there are only finitely many words of length less than or equal to
		$n$, clearly there is a $k \geq 1$ such that $v_1v_2\dots v_k \sim_n v$. 
		Let us denote $v_1v_2\dots v_k$ by $v'$. Note that since $e$ is an
		idempotent, $h(v') = e = h(v)$.
		Also by the induction hypothesis and the concatenation case already seen
		above, $h(v') = h(\hat{v'})$. By transitivity, $\hat{v'}$ is a finite
		subword of $v$ that is $\sim_n$ equivalent to it, and $h(v) = h(v') =
		h(\hat{v'})$. This proves the induction hypothesis for this case.

	 \item Case $v$ is an omega$^*$ sequence of words: This is symmetric to the case
		 above. The induction hypothesis follows from the following axiom of
		 \circmonoid: $\shuffleoper e = \omegastaroper {(\shuffleoper e)}$
	 \item Case $v$ is a "perfect shuffle" such that $h(v) = \shuffleoper{
		 \{b_1,\dots,b_k\}}$. By the "shuffle-power-trivial" property, we have 
		 $h(v) = (b_1 \productoper \dots \productoper b_k)^!$. Let $v = \prod_{i \in \rationals}
		 v_i$ where $h(v_i) \in \{b_1,\dots,b_k\}$. By induction hypothesis
		 $h(v_i) = h(\hat{v_1})$.  Since there are only finitely many words of
		 length less than or equal to $n$, there exists $l \geq k$ and $j_1,j_2,
		 \dots,j_l \in \rationals$ such that we get the following: $\{h(v_{j_1}),
		 h(v_{j_2}),\dots,h(v_{j_l})\} = \{b_1,\dots,b_k\}$ and
		 $v_{j_1}\dots v_{j_l} \sim_n v$.  Let $w = v_{j_1}\dots
		 v_{j_l}$.  Let $w^m$ be a finite power of $w$ such that $h(w^m) =
		 h(w)^!$.  Since the $\circmonoid$ satisfies the "$J$-trivial" identity $(x \productoper y)^! = (y \productoper x)^!$, 
		 we have that $h(w^m) = h(v)$ 
		 and $w^m$ is a finite subword of $v$ that
		 is $\sim_n$-equivalent to $v$. This
		 shows that the induction hypothesis also holds in this case.
 \end{enumerate}
Now for any two countable words $u$ and $v$, if $u \sim_n v$, then $h(u) = h(\hat{u}) =
h(\hat{v}) = h(v)$ where the middle equality is from the classical result of Simon
mentioned before.
Having shown that $u \sim_n v$ implies $h(u) = h(v)$, it now immediately follows that
$L$ is a boolean combination of $\sim_{n}$ equivalence classes and hence is recognized
by a morphism to $\cwq{\al}{n}$. 

%
%

$(2 \Rightarrow 1)$ Let $[x]_n J [y]_n$ for some $x,y \in \cw{\al}$.  There exists
$[u]_n, [v]_n, [u']_n, [v']_n$ such that $[u]_n \productoper [x]_n \productoper [v]_n = [y]_n$ and $[u']_n
\productoper [y]_n \productoper [v']_n = [x]_n$. In other words, $uxv \sim_n y$ and $u'yv' \sim_n x$.   
This implies $x \sim_n y$, that is, $[x]_n = [y]_n$. This proves $\cwq{\al}{n}$ is
$J$-trivial.

It is not difficult to see that $\shuffleword{\{u_1, \dots, u_p\}} \sim_n (u_1u_2 \dots
u_p)^n$. This means $\shuffleoper{\{[u_1]_n, \dots, [u_p]_n\}} = ([u_1]_n \dots
[u_p]_n)^!$.

$(2 \Rightarrow 3)$ Every equivalence class of $\sim_n$ is clearly definable in $\bc$.

$(3 \Rightarrow 2)$ Let $L$ be recognized by the formula $\alpha \defs \exists x_1, \dots,
x_n \varphi(x_1,\dots,x_n)$.  We show that for an $u \sim_n v$, $u \models \alpha$ if and
only if $v \models \alpha$. Consider an assignment $s$ which assigns the variables $x_i$s
to a position in the domain of $u$ such that $u, s \models \varphi$. Note that since
$\varphi$ is a quantifier free formula it is a boolean combination of formulas of the form
$x_i < x_j$, $x_i = x_j$ and $a(x_i)$. Let $X = \{s(x_i) \mid 1 \leq i \leq n\} \subseteq
\dom u$ be the set of $n$ points which are assigned to the $x_i$s. Since $u \sim_n v$,
there is a set $Y \subseteq \dom v$ of $n$ points such that $\subword{u}{X} =
\subword{v}{Y}$. Consider an assignment $\hat s$ to variables $x_i$ to positions in $Y$
such that $s(x_i) < s(x_j)$ iff $\hat s(x_i) < \hat s(x_j)$. Clearly such an assignment
satisfies $v, \hat s \models \varphi$ since the ordering between the variables and the
letter positions are preserved. Therefore we get that $u \models \alpha$ implies $v
\models \alpha$. A symmetric argument shows that other direction too. 

$(4 \Rightarrow 1)$ This is a trivial observation.

$(1 \Rightarrow 4)$ This follows from the fact that "identities" are preserved under  "division".
\qed \end{proof}

\subsection{Summary of FO subclasses}
\label{subsec:fo-summary}
We summarize the results of this section and known results from the literature. 

\begin{table}
\begin{tabular}{|c|c|c|}
\hline
Logic & Identity & Decomposition \\ \hline
$\foone$ & "shuffle-trivial" & $\times (U_1)$ \\
$\bc$ & "shuffle-power-trivial" & $\mathbf{S_n}$ \\
$\fotwo$ & $\ostar$-DA \cite{sav_fo2clo} & $wbp (U_1)$ \cite{DBLP:conf/lics/AdsulSS19} \\
$\fo$ & $\omegaoper x \productoper \omegastaroper x = x$ and shuffle simple \cite{colcombetCLO} & $bp(U_1)$ \cite{DBLP:conf/lics/AdsulSS19} \\ \hline
\end{tabular}
\end{table}


%% file: fo-inf.tex
\section{First Order Logic with "infinitary quantifiers"}
\label{sec:fo-inf}
\AP
Our results in the previous section resemble very closely the corresponding results over
finite words.  
This can be attributed to the limited capability of the operators
$\algomegasymb$, $\algomegastarsymb$ and $\algshufflesymb$ in the \circmonoids\ we
witnessed.  
%
As mentioned in the Introduction, \fo\ cannot define the language of infinite number of $a$'s.
An existential quantifier is a threshold counting quantifier - it says there exists at least one position satisfying a property. 
Using multiple such first-order quantifiers, \fo\/ can count up to any finite constant but not more.
Over "countable words", it is natural to ask for stronger threshold quantifiers. We introduce natural infinite extensions of the existential quantifier. These quantifiers can distinguish ordinals in the infinite. 
\newcommand{\Int}{\mathbb{Z}}
\newcommand{\inford}[1]{\mathcal{I}_{#1}}

\AP
%
We define $\inford 0$ to be the set of all non-empty finite orderings. For any number $n \in \nat$, we define the set $\inford n$ to be the set of all "orderings" of the form $\sum_{i \in \Int} \alpha_i$ where $\alpha_i \in \inford{n-1} \cup \{\emptyword\}$ and is closed under finite sum. We define the "Infinitary rank" (or simply \intro{rank}) of a "linear ordering" $\alpha$ (denoted by $\intro[\irank]{}\irank \alpha$) as the least $n$ (if it exists) where $\alpha \in \inford n$. If there is no such $n$ we say that the rank is infinite. For example, $\irank \omega = \irank {\omega+\omega} = \irank {\omega^* + \omega} = 1$, $\irank {\omega^2} = \irank{\omega^2 + \omega^*} = 2$, and the "rank" of  rational numbers is infinite. 
%
%

\AP
We introduce the logic $\intro{\foinf}$ extending \fo\ with ""infinitary quantifiers""\intro[\ei]{}
: $\ei 0 x ~\varphi \mid \ei{1} x ~\varphi \mid \ldots \mid \ei{n} x ~\varphi 
\mid \ldots$ for all $n \in \nat$.

\setlength{\intextsep}{0pt} %
\setlength{\columnsep}{0pt}%
\begin{wrapfigure}{r}{3cm}
		\centering
		\begin{tikzpicture}
		\node at (0,0) {$\unit$};
		\draw[thick] (-0.3, -0.3) rectangle (0.3, 0.3);
		\node at (0,-0.7) {$0$};
		\draw[thick] (-0.3, -1) rectangle (0.3, -0.4);
		\end{tikzpicture}
\captionsetup{labelformat=empty}
	\caption*{$\ali 0$-chain}
\end{wrapfigure}
Note that all the variables are first order. The semantics of the "infinitary
quantifier" $\ei{n} x$ for an $n \geq 0$ is: for a word $w$ and an assignment $s$, we say $w,s \models
\ei{n} x~ \varphi$ if there exists a "subordering" $X \subseteq \dom{w}$ such that
$\irank{X} = n$  and $w, s[x=i] \models \varphi$ for
all $i \in X$.  For example, $\ei 0 x ~\varphi$ is equivalent to $\exists x ~\varphi$ since both formulas are true if and only if there 
is at least one satisfying assignment $x=s$. 

\begin{wrapfigure}{r}{3cm}
		\centering
		\begin{tikzpicture}
		\node at (0,0) {$\unit$};
		\draw[thick] (-0.3, -0.3) rectangle (0.3, 0.3);
		\node at (0,-0.7) {$0$};
		\draw[thick] (-0.3, -1) rectangle (0.3, -0.4);
		\node at (0,-1.4) {$1$};
		\draw[thick] (-0.3, -1.7) rectangle (0.3, -1.1);
		\node at (0,-1.9) [circle, fill, inner sep=1pt] {};
		\node at (0,-2.1) [circle, fill, inner sep=1pt] {};
		\node at (0,-2.6) {$n$};
		\draw[thick] (-0.3, -2.9) rectangle (0.3, -2.3);
		\end{tikzpicture}
	\caption{$\ali n$-chain}
	\label{fig:chain-algebra}
	\vspace{1cm}
\end{wrapfigure}

\AP
The logic $\intro[\foi]{}\foi{n}$ denote the fragment containing only the infinitary quantifiers $\ei{j} x$ for all $j \leq n$.
Clearly the following relationship is maintained among the logics: 
\[\fo = \foi 0 \subseteq \foi{1} \subseteq \foi{2} \subseteq\ldots \]
\AP
We also denote by $\intro[\foione]{}\foione n$ the corresponding one variable fragment of $\foi n$.
\begin{example}
	The formula $\ei{1}x~a(x)$ denotes the set of all countable words with infinitely
	many positions labelled $a$. Since $\fo$ cannot express this, it shows
	$\fo \subsetneq \foi{1}$.
\end{example}

\AP
For an $n \geq 0$, we define the \circalgebra\ $\intro[\ali]{}\ali{n}$-chains as: $\algebraOf {\{\unit, 0, 1, \dots, n \}}$ where 
for all $0 \leq i \leq j \leq n$, $i \productoper j = j \productoper i = \max(i,j)= j$ and for all  $0 \leq k < n$, $\omegaoper k = \omegastaroper k = k+1$ and $\omegaoper n = \omegastaroper n = n$. That is, $\omegaoper k = \omegastaroper k = \min(k+1,n)$
Moreover, $\shuffleoper \unit = \unit$ and $\shuffleoper S = n$ for any $S$ where $S \backslash \{\unit\} \neq \emptyset$. 
\[ \ali{n} \defs (\{\unit, 0, 1, \ldots, n\}, \{i,j\} \xmapsto[]{\productoper} \max(i,j), i \xmapsto[]{\algomegasymb} \min(i+1,n), i \xmapsto[]{\algomegastarsymb}\min(i+1,n), S \xmapsto[]{\algshufflesymb}n ) 
\footnote{As mentioned in the Preliminaries, we restrict the descriptions of "derived operators" to $\ali n \setminus \{\unit\}$} 
\] 
Note that the syntactic \circalgebra\ for the language defined by $\ei n x ~a(x)$ is $\ali n$.

\subsection{$\foinf$ with single variable}
\label{subsec:block-fo-inf}
\AP
In this section we show that languages recognized by $\ali{n}$ are definable in $\foione {n}$. It is easy to observe that the "direct product" of $\ali{n}$ recognize exactly those "languages" definable in the one variable fragment. 
\AP
\begin{theorem}\label{thm:singledelta}
"Languages" recognized by "direct product" of $\ali n$ are exactly those definable in $\foione n$.
\end{theorem}
\begin{proof}
\AP
 We first show that "languages" recognized by $\ali n$ are definable in $\foione {n}$. Let $h \colon \words{\al} \to \ali{n}$ be a "morphism". 
 It suffices to show that for any element $m \in \ali{n}$, $h^{-1}(m)$ is definable in $\foione {n}$. In the rest of the discussion we adopt the convention that $\unit < 0$. 
Let ${\uparrow}{m}$ denote the set $\{m' \mid m' \geq m\}$. Note that for an $m < n$, $h^{-1}(m) = h^{-1}({\uparrow}m) \setminus h^{-1}({\uparrow}(m+1))$ and $h^{-1}(n) = h^{-1}({\uparrow}n)$.  
Therefore, it is sufficient to show that $h^{-1}({\uparrow}m)$ is definable in $\foione {n}$. 
\AP
	For each $m \in \ali n$, we define the "language" $L(m)$ as \\
	\noindent $\big\{ w \mid $ there exists a letter $a$ in $w$ such
	that $h(a) = j \neq \unit$ and either $j \geq m$ or there is a set of positions $\alpha$ labelled $a$ such that $\irank{\alpha} = j'$ and 
	$j+j' \geq m \big\}$ \\
\AP
	\noindent The following $\foione {n}$ sentence defines the "language" $L(m)$. 
	\[
	\bigvee_{\substack{a \in \al,~ h(a) \geq m}} \exists x ~a(x) ~~\vee~ \bigvee_{\substack{a \in \al,~ 0 \leq h(a) < m}} \ei{m-h(a)} x ~a(x)
	\]
\AP
\noindent We show that $L(m) = h^{-1}({\uparrow}m)$ by induction on $m$. The base case holds since 
	${\uparrow}\unit = \ali{n}$, $h^{-1}({\uparrow}\unit) = \words \al$ and $L(\unit) = \words \al$. To prove the induction hypothesis assume the claim holds for all $j < m$. Consider a "word" $w$. By a second induction on the height of an "evaluation tree" $(T,h)$ for $w$ we show for all words $v \in T$, $v \in h^{-1}({\uparrow}m)$ if and only if $v \in L(m)$. In each of the following cases we assume that the children of the node (if they exist) satisfy the second induction hypothesis.
%

\begin{enumerate}
\item Case $v$ is a letter: The hypothesis clearly holds 
\item Case $v$ is a "concatenation" of two words $v_1$ and $v_2$: There are two cases to consider - $\{v_1,v_2\} \cap h^{-1}({\uparrow}m) \neq \emptyset$ or not. In the first case, let for an $i \in \{1,2\}$ we have $h(v_i) \geq m$ and $v_i \in L(m)$. Clearly $h(v) = h(v_1v_2) \geq m$ and $v \in L(m)$. For the second case, let us assume $h(v_1) = i$ and $h(v_2)=j$ such that $i \leq j < m$ and both $v_1,v_2 \notin L(m)$. From the definition of $\ali n$, it follows that $h(v) = h(v_1v_2) =j$. Let the $a$-labelled "suborderings" in $v_1$ and $v_2$ be $\alpha_1$ and $\alpha_2$ respectively where $\irank{\alpha_1} \leq \irank{\alpha_2} = j'$. It follows from the definition that $\irank{\alpha_1+\alpha_2} = j'$ and therefore $v \notin L(m)$. 
\item Case $v$ is an "omega" sequence of words $\langle v_1, v_2, \dots, \rangle$ such that $h(v_i) = k$, for all $i$, and $k$ is an "idempotent" (in $\ali n$ all elements are "idempotents"): Firstly, if $k \geq m$ and $v_i \in L(m)$ then clearly $h(v) \geq m$ and $v \in L(m)$. The non-trivial case is $k = m-1$. From the second induction hypothesis $v_i \notin L(m)$ for all $i$. From the definition of $\ali n$, $h(v) =\omegaoper k = m$. We need to show that $v \in L(m)$. By first induction hypothesis, each $v_i$ has a letter $a_i$ and an $a_i$-labelled set of positions $\alpha_i$ such that 
			 $h(a_i) + \irank{\alpha_i} = k$. Since $|\al|$ is finite, there is a letter $a$ occurring in "omega"
			many factors. Hence the $a$-labelled set of positions $\alpha$ in $v$ satisfies $h(a)
			+\irank{\alpha} = k+1$ or in other words $v \in L(m)$. 
\item Case $v$ is an "omega$^*$" sequence: This case is symmetric to the above case.
\item Case $v$ is a "perfect shuffle", $h(v) = \shuffleoper{S}$: It is easy to see that the induction hypothesis holds if $S=\{\unit\}$. So, assume $S \cap \{\unit\} \neq \emptyset$. Hence $h(v) = n$. Since, there are rational number of children $u$ where $h(u) \neq \unit$, there is a letter $a$ such that $a$-labelled set of positions in $v$ has infinite "rank" or $v \in L(n)$. 
\end{enumerate}
\AP

Now we give the proof of the other direction. Due to the restriction of a single variable, any quantifier free formula $\varphi(x)$ is a boolean combination of atomic letter predicates. Since any position has exactly one "letter", we can consider $\varphi(x)$ to be a disjunction of the "letter" predicates. Consider the formula $\alpha \defs \ei k x \vee_{a \in A} a(x)$. The language "defined" by $\alpha$ is "recognized" by the morphism $h: \words \alphabet \mapsto \ali n$ where $h(a) = 0$ for all $a \in A$ and $h(a) = \unit$ for all $a \notin A$ since $h^{-1}(\{k,k+1,\dots,n\})$ "recognizes" all words where the $A$-labelled positions have "rank" at least $k$. We conclude by observing that
"languages" defined by boolean combinations of such sentences can be "recognized" by "direct products" of $\ali n$.
\qed \end{proof}

\subsection{The general $\foinf$ logic}
\AP
In this section, we consider the full logic $\foi{n}$ and observe that they define exactly those "languages" "recognized" by block products of $\ali{n}$. 
\begin{theorem}\label{thm:fo-inf-n}
The "languages" defined by $\foi{n}$ are exactly those "recognized" by finite "block products" of $\ali{n}$. 
Moreover, the "languages" defined by $\foinf$ are exactly those "recognized" by finite "block products" of $\{\ali{n} \mid n \in \nat\}$.
\end{theorem}
\begin{proof}
We first show that "languages" "recognizable" by finite "block products" of $\ali n$ are definable in $\foi n$.
The proof is via induction on the number of $\ali{n}$ in an
	iterated "block product". The base case follows from Theorem~\ref{thm:singledelta}.

	For the inductive step, consider a "morphism" $h \colon \words{\al} \to M \blockproduct \ali{n}$.
	Let $h_1 \colon \words{\al} \to M$ be the induced morphism to $M$, and let $\sigma$
	be the associated transducer. By the "block product principle" (see Proposition \ref{prop:bpp}), any language
	"recognized" by $h$ is a boolean combination of "languages" $L_1 \subseteq \words{\al}$
	"recognized" by $M$ and $\sigma^{-1}(L_2)$ where $L_2 \subseteq \words{(M \times \al
	\times M)}$ is "recognized" by $\ali{n}$. By induction hypothesis, $L_1$ is
	$\foi{n}$ definable. By the base case $L_2$ is $\foi{n}$ definable but over
	the alphabet $M \times \al \times M$. To complete the proof, one needs to show for any
	word $w \in \words{\al}$ and assignment $s$, and for any $\foi{n}$ formula $\varphi$
	over the alphabet $M \times \al \times M$, there exists a $\foi{n}$ formula
	$\hat{\varphi}$ over the alphabet $\al$ such that $w, s \models \hat{\varphi}$ if
	and only if $\sigma(w) , s \models \varphi$.  For instance, suppose $\varphi =
	\ei{i}x~(m_1, c, m_2)(x)$, and inductively $\phi_{m_1}$ (resp. $\phi_{m_2}$) are
	$\foi{n}$ formula characterizing "words" over $\words{\al}$ that are mapped by $h_1$ to
	$m_1$ (resp. $m_2$). Then $\hat{\varphi}$ is $\ei{i}x~(\phi_{m_1}|_{<x} \land c(x)
	\land \phi_{m_2}|_{>x})$, where $\phi_{m_1}|_{<x}$ is the formula $\phi_{m_1}$
	with all its variables relativised to less than the variable $x$.  This way, one
	proves that $\sigma^{-1}(L_2)$ is $\foi{n}$ definable. This completes the proof of this direction.
	
We now show the other direction: a "language" defined by an arbitrary formula $\varphi \in \foi{n}$ is "recognized" by finite "block products" of $\ali n$.
For any $\foi{n}$ formula, we can naturally consider its
	models $w,s$ as words over extended alphabets; for each free variable an $1$ or $0$
	is added to the label of a position in $w$ depending on whether $s$ assigns the
	corresponding free variable to that position or not. The proof now goes via structure induction on the subformulas $\phi$ of $\varphi$.

	Case $\phi = a(x)$: Then $L(\phi) \subseteq \words{(A \times \{0,1\})}$ is
	the set of "words" with a unique position labeled $(a,1)$. This can be
	"recognized" by $U_1 \blockproduct U_1$ \cite{DBLP:conf/lics/AdsulSS19}. Since $U_1$ divides $\ali{n}$, the hypothesis holds.
 
 	Case $\phi = x <y$: This can also be "recognized" by "block products" of $U_1$ \cite{DBLP:conf/lics/AdsulSS19}. 

	Boolean combination of formula can be handled by "direct product" of inductively defined \circalgebras.

	Case $\phi = \ei{i}x ~\psi$ (for $i \leq n$): This is the non-trivial case. Let $L(\psi) \subseteq
	\words{(\al \times \{0,1\})}$ is inductively "recognized" by $h
	\colon \words{(\al \times \{0,1\})} \to M \in \bpc \ali{n}$, that is, there
	is a set $F \subseteq M$ such that $h^{-1}(F) = L(\psi)$. We prove that $M
	\blockproduct \ali{n}$ recognizes $L(\phi)$. Once again we use the "block product
	principle".  Consider two "morphisms" $g_1 \colon \words{\al} \to M$ and $g_2
	\colon \words{(M \times \al \times M)} \to \ali{n}$. Let $g_1(a) = h((a,0))$ and
	suppose $g_2((m_1,a,m_2))$ equals $0$ if $m_1 \productoper h((a,1)) \productoper m_2 \in F$, and
	it equals $\unit$ otherwise. Let $\sigma$ be the transducer corresponding to
	$g_1$. We show that $w \models \phi$ if and only if
	$g_2(\sigma(w)) \geq i$. This would imply $L(\phi) = \sigma^{-1}(g_2^{-1}(\{i,
	i+1, \ldots, n\}))$ and by the "block product principle", this is recognized by $M
	\blockproduct \ali{n}$. 

	Let $w \models \phi$. If $\alpha_\psi$ is the set of all positions
	of $w$ where $\psi$ is true, then $\irank{w_\psi} \geq i$. Let $l \in \alpha_{\psi}$ and $w(l)=a$.  
	We can split $w$ at the position $l$ as $w_1 a w_2$ and by logic semantics 
	$w_1^0 (a, 1) w_2^0 \models \psi$ (for any $u \in \words{\al}$, 
	we denote by $u^0$ the word over the same domain with $u^0[i] = (u[i], 0)$). If
	$h(w_1^0) =m_1$ and $h(w_2^0) = m_2$, then $m_1 \productoper h((a,1)) \productoper m_2 \in F$.
	Also, $\sigma(w)[l] = (m_1, a, m_2)$. So, $g_2$ maps every position $l \in
	\alpha_{\psi}$ to $0$, and hence $g_2(\sigma(w)) \geq i$. Conversely,
	suppose $g_2(\sigma(w)) \geq i$. Let $\alpha_0$ denote the positions of
	$\sigma(w)$ for which $g_2$ maps to $0$. Since $g_2$ maps
	each letter to $0$ or $\unit$, we get $\irank{\alpha_0} \geq i$. Let $l \in
	\alpha_0$. If $\sigma(w)(l) = (m_1, a, m_2)$, then
	$m_1 \productoper h((a,1))\productoper  m_2 \in F$. This means $\psi$ is true at position $l$ for $w$.
	Since $l$ is any position in $\alpha_0$, we have that $w \models \phi$. 
\qed \end{proof}
\AP
%
\newcommand{\cut}[1]{\hat #1}
\AP
We claim that both first order logic with "cuts" ($\intro{\focut}$) and weak monadic second order logic ($\intro{\wmso}$) can define the languages definable in $\foinf$.
\begin{theorem}\label{thm:foinf-focut-wmso}
	 $\foinf \subseteq \focut \cap \wmso$ \footnote{Here, $\foinf$, $\focut$, $\wmso$ denote the languages defined by the respective logic.} 
\end{theorem}
\begin{proof}
We first show by structural induction that there is an equivalent $\wmso$ formula for any $\foinf$ formula. It is easy to observe that the hypothesis holds for the atomic case, first order quantification and boolean combinations. Let us consider the formula $\phi = \ei k x ~\psi(x)$. By our inductive hypothesis there is a $\wmso$ formula $\hat \psi(x)$ equivalent to $\psi(x)$. We show that the $\wmso$ formula $\Psi_k$ inductively defined is equivalent to $\phi$: Let $\Psi_0 \defs \exists x ~\hat \psi(x)$ and 
\begin{align*}
& \Psi_n \defs \text{``For any finite set $X = \{x_1,\dots,x_k\}$, one of the "intervals" $[-,x_1],\dots,$ } \\
& \text{$[x_i,x_{i+1}],\dots,[x_k,-]$ can be split into at least two parts each satisfying $\Psi_{n-1}$''}
\end{align*}
This can be expressed in $\wmso$ as follows (consec$(X,x,y)$ says that $x,y \in X$ and $x< y$ and there is no $z \in X$ such that $x<z< y$. That is $x$ and $y$ are consecutive in set $X$): 
\begin{align*}
\Psi_n \defs & \forall_{fin} X ~\Big( \exists x,y \in X ~\exists z (\text{consec}(X,x,y) \ifthen \Psi_{n-1}[>x,<z] \wedge \Psi_{n-1}[>z,<y]) ~\vee \\
& \exists z ~\big(\Psi_{n-1}[>z,< \text{min}(X)] \wedge \Psi_{n-1}[< z]\big) ~\vee~ \exists z ~\big(\Psi_{n-1}[> \text{max}(X), < z] \wedge \Psi_{n-1}[> z]\big) \Big)
\end{align*}
We claim that $\Psi_n$ satisfies all words where the $\psi$-labelled set of positions $\alpha$ has $\irank \alpha \geq n$. It is clearly true for the base case $\Psi_0$. Assume the hypothesis is true for all $j< n$. The formula $\Psi_n$ says that for any finite number of partitions $\alpha_1, \alpha_2, \dots, \alpha_k$ of the $\psi$-labelled set of positions $\alpha$, there is at least one $\alpha_i$ that can be split into two parts containing $\psi$-labelled set of positions $\alpha^1_i$ and $\alpha^2_i$ such that $\irank {\alpha^1_i} \geq n-1$ and $\irank {\alpha^2_i} \geq n-1$. In short, finite partitioning of $\psi$-labelled set of positions with rank $n-1$ is not possible or $\irank \alpha \geq n$. Therefore the formula $\Psi_k$ is equivalent to the formula $\phi$.

Next we give an $\focut$ formula equivalent to an $\foinf$ formula. Like in the previous proof, let us look at the case $\phi = \ei k x ~\psi(x)$ where $\psi(x)$ is equivalent to an $\focut$ formula $\hat \psi(x)$. We show $\phi$ is equivalent to $\Phi_k$ where $\Phi_n$ is inductively defined as: $\Phi_0 \defs \exists x ~\hat \psi(x)$ and $\Phi_n$ is 
\[
\text{``There is a "cut" towards which there is an "omega" (or "omega$^*$") sequence of "factors" satisying $\Phi_{n-1}$''}
\]
This can be written in $\focut$ as follows:
\begin{align*}
~\exists_{cut} \cut x ~\forall y < \cut x ~\exists z \in (y,\cut x)  ~\Phi_{n-1} [>y,<z] \big)  ~\vee ~\exists_{cut} \cut x ~\forall y > \cut x ~\exists z \in (\cut x, y) ~\Phi_{n-1}[>z,<y]
\end{align*}
The formula says there is an "omega" or "omega$^*$" sequence of $\psi$-labelled positions which are of "rank" $n-1$. 
We now argue that the inductively defined formulas $\Phi_n$ says there is a $\psi$-labelled set of positions $\alpha$ where $\irank \alpha \geq n$. The base is same as the previous case. Let us consider the formula $\Phi_n$. We argue that the first (resp. second) disjunct says there is an "omega" (resp. "omega$^*$") sequence of "domains" $\alpha_1, \alpha_2, \dots, $ such that each of the $\alpha_i$'s are $\psi$-labelled positions and $\irank {\alpha_i} \geq n-1$. The first (resp. second) disjunct guesses a "cut" towards which this "omega" (resp. "omega$^*$") sequence approaches. It then says that for any point strictly before (resp. after) this "cut", there is a point strictly after (resp. before) between which there is a $\psi$-labelled set of positions $\alpha$ where $\irank{\alpha} \geq n-1$. Clearly, this implies the $\psi$-labelled set of positions have "rank" $\geq n$. We conclude by observing that the $\focut$ formula $\Phi_k$ is equivalent to the formula $\phi$.
\qed \end{proof}

%% file: nofinbasis.tex
The main goal of this section is to prove that $\foinf, \focut$ and $\focut \cap \wmso$ over "countable
words" do not admit a "block product" based characterization which uses only a {\em finite} set
of \circmonoids. This is in stark contrast with the result in \cite{DBLP:conf/lics/AdsulSS19} 
which shows that a language of "countable words" is \fo-definable iff it is recognized
by a strong iteration of "block product" of copies of $\ali{0}$ (alternately called $\uone$). This is abbreviated by saying that
\fo\ has a block-product based characterization using a basis which contains the single
\circmonoid\ $\ali{0}$. 
Notice that, it follows from the results in the previous section that
\foinf\ admits a "block product" based characterization using the {\em natural infinite
basis} $\{\ali{n}\}_{n \in \nat}$.


Fix a finite \circalgebra\  $\monoid = \algebraOf \monoidset$. For every $n \in \nat$, we define
the operation $\ogi{n}: \monoidset \to \monoidset$ which maps $x$ to $x^{\ogi{n}}$.
The inductive definition of $\ogi{n}$ is as follows:
$x^{\ogi{0}} = x^!$ and
$x^{\ogi{n}} = (\omegaoper{(x^{\ogi{n-1}})}\omegastaroper{(x^{\ogi{n-1}})})^!$. 

\begin{lemma} \label{lem:fingapnest}
For each $m \in \monoidset$, there exists $n$ such that
$\forall n' \geq n, m^{\ogi{n}} = m^{\ogi{n'}}$.
\end{lemma}
\begin{proof}
Let $m$ be an arbitrary element in a $\circmonoid\ \monoid$. We show that there exists an
$n$ such that for all $n' > n$ we have $m^{\ogi{n}} = m^{\ogi{n'}}$.  Consider the
following sequence: $a_0 = m$ and $a_{j+1} = (\omegaoper{(a_j)} \productoper
\omegastaroper{(a_j)})^!$. Note that $m^{\ogi j} = a_j$ for all $j$. We argue that for
consecutive $a_j, a_{j+1}$ either $a_{j} \Jg a_{j+1}$ or $a_{j+1} = a_j$. Let us assume
$a_{j} \Jeq a_{j+1}$. Clearly $a_{j+1}$ is $\Req$ and $\Leq$ equivalent to $a_j$.
Therefore $a_j \Heq a_{j+1}$. Since in a $\Jeq$ class containing a group $a_j \nJeq
\omegaoper{a_j}$ we have that the $\Heq$ class of $a_j$ has cardinality one and therefore
$a_j = a_{j+1}$.
\qed \end{proof}

\AP
We now define the \intro{gap-nesting-length} of $\monoid$ (in notation, $\intro[\gnl]{}\gnl{\monoid}$) to be 
the smallest $n$ such that for all $m \in M$, $m^{\ogi{n}} = m^{\ogi{n+1}}$.
It follows from the previous lemma that a finite \circalgebra\ has a finite
"gap-nesting-length". It is a simple computation that,
for each $k$, $\gnl{\ali{k}}=k$. 
The following main technical lemma is the key to our no-finite-basis theorems.
\begin{lemma}\label{main-lemma} For finite aperiodic\footnote{This means the underlying
	monoid of a \circalgebra\ is aperiodic} \circalgebras\ \monoid\ and \monoidN\ , 
	\begin{enumerate}
	\item We have, $\gnl{\monoid \blockproduct \monoidN} \leq
			\max{(\gnl{\monoid},\gnl{\monoidN})}$. 
	\item If $\monoid$ divides $\monoidN$ then $\gnl{\monoid} \leq \gnl{\monoidN}$.
	\end{enumerate}
\end{lemma}

Before proving Lemma~\ref{main-lemma}, we first recall the product and omega operations of
semidirect products from \cite{DBLP:conf/lics/AdsulSS19}.  Consider two \circalgebras\ 
$\monoid = \algebraOf M$ and $\monoidN = (N, \widehat {id}, + , \widehat{\tau},
\widehat{\tau}^*, \widehat{\kappa})$. We denote the left and right action of $M$ on $N$ by $*$.
The semidirect product is $\monoid \ltimes \monoidN = ( M \times N, \tilde{\cdot},
\tilde{\tau}, \tilde{\tau}^*, \tilde{\kappa} )$ where the operations are suitable defined.
We recall the product operation and the omega and omega$^*$ operation for an element
$(e,n)$ where $e$ is an idempotent.
\begin{enumerate}
	\item $(m_1, n_1) ~\tilde{\cdot}~ (m_2, n_2) = (m_1 \cdot m_2,  n_1 \ast m_2 ~+~ m_1 \ast n_2)$
	\item $(e, n)^{\tilde{\tau}} = (\omegaoper e, n \ast \omegaoper e ~+~ (e \ast n \ast \omegaoper e)^{\widehat{\tau}})$ 
	\item $(e, n)^{\tilde{\tau}^*} = (\omegastaroper e, (\omegastaroper e \ast n \ast e)^{\widehat{\tau}^*} ~+~ \omegastaroper e \ast n)$ 
\end{enumerate}
The block product of $\monoid$ and $\monoidN$, denoted $\monoid \blockproduct \monoidN$,
is the semidirect product of $\monoid$ and $\monoidK$ where $\monoidK$ is the direct
product of $|M| \times |M|$ copies of $\monoidN$. Note that the underlying set of
$\monoidK$ can also be considered as $N^{M \times M}$: the set of all functions from $M 
\times M$ to $N$. We refer to \cite{DBLP:conf/lics/AdsulSS19} for action axioms over \circalgebra.

\begin{lemma}\label{lem:id-bp}
	Let $(m,f) \in \monoid \blockproduct \monoidN$ be an idempotent. Then $m \in M$ is
	an idempotent and $m \ast f \ast m$ is an idempotent in $N^{M \times M}$.
\end{lemma}
\begin{proof}
	Since $(m,f)$ is an idempotent, by concatenation rule of block product algebra, we
	have $(m,f) = (m^2, f \ast m + m \ast f)$. Hence $m = m^2$, that is, $m \in M$ is
	an idempotent. Also, $f = f \ast m + m \ast f$ implies $m \ast f \ast m = m \ast f
	\ast m + m \ast f \ast m$. Hence, $m \ast f \ast m$ is an idempotent in $N^{M
	\times M}$.
\qed \end{proof}

\begin{lemma}\label{lem:idp-bp}
	Let $(m,f), (m',f') \in \monoid \blockproduct \monoidN$ such that $(m,f) = (m', f')^!$. Then $m
	= (m')^!$. Further, if $M$ is aperiodic, then $m \ast f \ast m = (m \ast
	f' \ast m)^!$.
\end{lemma}
\begin{proof}
	Suppose $k$ is an idempotent power of $(m',f')$. So $m = (m')^k$ and $f = \sum_{i
	=0}^{k-1} (m')^i \ast f' \ast (m')^{k-i-1}$. By Lemma~\ref{lem:id-bp}, $m$ is an
	idempotent, so $m = (m')^!$.

	If $M$ is aperiodic, then $(m')^j = m$ for $j \geq k$. Hence $m \ast f \ast m = (m
	\ast f' \ast m)^k$. By Lemma~\ref{lem:id-bp}, $m \ast f \ast m$ is an idempotent.
	So $m \ast f \ast m = (m \ast f' \ast m)^!$.
\qed \end{proof}

\begin{lemma}\label{lem:action-gamma}
	Consider \circalgebra\ $\monoid$ has compatible left and right actions on
	\circalgebra\ $\mathbf{P}$. 
	Let $m,m' \in M$ and $p \in P$.
	Then  $m \ast p^{\gamma_n} \ast m' =
	(m \ast p \ast m')^{\gamma_n}$ 
\end{lemma}
\begin{proof}We first prove that $m \ast p^! \ast m' = (m \ast p \ast m')^!$.
By action axioms for concatenation, it is easy to see that $m \ast p^k \ast m' = (m \ast p
\ast m')^k$ for any natural number $k \geq 1$. Note that there exists $k \in \nat$ such
that $p^k = p^!$ and $(m \ast p \ast m')^k = (m \ast p \ast m')^!$. Then 
$m \ast p^! \ast m' = m \ast p^k \ast m' = (m \ast p \ast m')^k = (m \ast p \ast
m')^!$.  

The proof is now by induction on $n$. For $n=0$, we have $m \ast
p^{\ogi{0}} \ast m = m \ast p^! \ast m = (m \ast p \ast m)^! = (m \ast p \ast
m)^{\ogi{0}}$.

For the inductive step, note that
\begin{align*}
	m \ast p^{\gamma_n} \ast m' &= m \ast (\omegaoper{(p^{\ogi{n-1}})} \productoper
	\omegastaroper{(p^{\ogi{n-1}})})^! \ast m'  & \text{defn. of $\ogi{n}$} \\
						    &= (m \ast
	(\omegaoper{(p^{\ogi{n-1}})} \productoper \omegastaroper{(p^{\ogi{n-1}})}) \ast
	m')^!  & \text{ } \\
	       &= ((m \ast \omegaoper{(p^{\ogi{n-1}})} \ast m') \productoper (m
	\ast \omegastaroper{(p^{\ogi{n-1}})} \ast m'))^!   & \text{action axiom for
		$\productoper$} \\
	&= (\omegaoper{(m \ast {(p^{\ogi{n-1}})} \ast m')} \productoper 
		\omegastaroper{ (m \ast {(p^{\ogi{n-1}})} \ast m')})^! & \text{action
			axiom for $\algomegasymb$, $\algomegastarsymb$} \\
	&= (\omegaoper{((m \ast p \ast m')^{\ogi{n-1}})} \productoper 
			\omegastaroper{((m \ast p \ast m')^{\ogi{n-1}})})^! & \text{
				induction hypothesis} \\
				    &= (m \ast p \ast m')^{\ogi{n}} & \text{defn. of
				    $\ogi{n}$}
\end{align*}
This completes the proof.
\qed \end{proof}

\begin{lemma}\label{lem:gamma-bp}
	Consider $(m,f), (m',f') \in \monoid \blockproduct \monoidN$ such that $(m,f) =
	(m',f')^{\ogi{n}}$. Then $m = (m')^{\gamma_n}$. If
	$\monoid$ is aperiodic, then
	$m \ast f \ast m = (m \ast f' \ast m)^{\ogi{n}}$.
\end{lemma}
\begin{proof}
	The proof is by induction on $n$. For the base case of $n=0$, we have $(m,f) =
	(m',f')^{\ogi{0}} = (m',f')^!$. By Lemma~\ref{lem:idp-bp}, $m = (m')^! =
	m^{\ogi{0}}$ and if $\monoid$ is aperiodic, $m \ast f \ast m = (m \ast f' \ast m)^! = (m \ast f' \ast
	m)^{\ogi{0}}$. This proves the base case.

	For the inductive step, let $(m,f) = (m',f')^{\ogi{n}} =
	((m',f')^{\ogi{n-1}})^{\ogi{1}}$. Also let $(e,g) = (m',f')^{\ogi{n-1}}$. So $(m,f)
	= (e,g)^{\ogi{1}}$. By induction hypothesis, $e = (m')^{\ogi{n-1}}$ and $m =
	e^{\ogi{1}}$ implying $m = ((m')^{\ogi{n-1}})^{\ogi{1}} = (m')^{\ogi{n}}$.  If
	$\monoid$ is aperiodic, then by induction hypothesis, $e \ast g\ast e = (e \ast f'
	\ast e)^{\ogi{n-1}}$ and $m \ast f \ast m = (m \ast g \ast m)^{\ogi{1}}$. Note
	that since $m = e^{\ogi{1}} = (\omegaoper{e} \productoper \omegastaroper{e})^!$,
	we have $m \productoper e = e \productoper m = m$. Therefore
	\begin{align*}
		m \ast f \ast m &= (m \ast g \ast m)^{\ogi{1}} \\
				&= (m \ast (e \ast g \ast e) \ast m)^{\ogi{1}} \\
				&= (m \ast (e \ast f' \ast e)^{\ogi{n-1}} \ast
				m)^{\ogi{1}} \\
				&= ((m \ast f' \ast  m)^{\ogi{n-1}})^{\ogi{1}} \\
				&= (m \ast f' \ast m)^{\ogi{n}}
	\end{align*}
	This completes the proof.
\qed \end{proof}

\begin{proof} [of Lemma \ref{main-lemma}]

We first prove the first case. Consider two aperiodic $\ostar$-monoids $\monoid$ and $\monoidN$ such that 
	$\max{(\gnl{\monoid}, \gnl{\monoidN})} = k \in \nat$. We show that $\gnl{\monoid
	\blockproduct \monoidN} \leq k$.  Note that, for any $m \in \monoid$ and any $n
	\in \monoidN$, $m^{\ogi{k}} = m^{\ogi{k+1}}$ and $n^{\ogi{k}} = n^{\ogi{k+1}}$.  

	Let $(m,f) \in M \blockproduct N$ be an arbitrary element. We show that
	$(m,f)^{\ogi{k}} = (m,f)^{\ogi{k+1}}$. Let $(e,g) = (m,f)^{\ogi{k}}$. Then
	$(e,g)^{\ogi{1}} = (m,f)^{\ogi{k+1}}$. Also by Lemma~\ref{lem:gamma-bp}, $e =
	m^{\ogi{k}}$ and $e \ast g \ast e = (e \ast f \ast e)^{\ogi{k}}$.  Since $\monoid$
	and $\monoidN$ have gap-nesting-length less than or equal to $k$, we get
	$e = m^{\ogi{k}} = m^{\ogi{k+1}} = e^{\ogi{1}}$ and $e \ast g \ast e = (e \ast f
	\ast e)^{\ogi{k}} = (e \ast f \ast e)^{\ogi{k+1}} = (e \ast g \ast e)^{\ogi{1}}$.
	Now we use the fact that in any \circalgebra\ $x = x^{\ogi{1}}$ implies $x J
	\omegaoper x$ and $x J \omegastaroper x$ and that further implies $x = \omegaoper
	x \productoper \omegastaroper x$. See~\cite{colcombetCLO} for the details
	regarding the simple proof of this
	property based on Green's relations.

	Therefore we have $e = \omegaoper e \productoper \omegastaroper e$ and $e \ast g
	\ast e = \omegaoper {(e \ast g \ast e)} + \omegastaroper {(e \ast g \ast e)}$.
	Since $(e,g)$ is an idempotent by definition of the $\ogi{i}$ operation, we get
	that $e$ is an idempotent in $\monoid$ by Lemma~\ref{lem:id-bp}. Therefore
	\begin{align*}
		&\omegaoper{(e,g)} \productoper \omegastaroper{(e,g)} \\
		&= (\omegaoper e \omegastaroper e , g \ast \omegaoper e \omegastaroper e + 
		\omegaoper {(e \ast g \ast \omegaoper e \omegastaroper e)} +
		\omegastaroper {(\omegaoper e \omegastaroper e \ast g \ast e)} +
		\omegaoper e \omegastaroper e \ast g) \\
		&= (e, g \ast e + \omegaoper{(e \ast g \ast e)} + \omegastaroper{(e \ast g
		\ast e)} + e \ast g) \\
		&= (e, g \ast e + e \ast g \ast e + e \ast g) \\
		&= (e,g)^3 \\
		&= (e,g)
	\end{align*}
	Hence $(m,f)^{\ogi{k+1}} = (e,g)^{\ogi{1}} = (e,g) = (m,f)^{\ogi{k}}$. This
	completes the proof for the block product operation.

 Now we prove the second case. If $\monoid$ is a subalgebra of $\monoidN$, then the property is
	easily verified. Let's suppose $h \colon \monoidN \to \monoid$ is a surjective
	morphism, and $\gnl{\monoidN} = k$. For any $m \in \monoid$, there exists $n \in 
	\monoidN$ such that $h(n) = m$. It is straightforward to check that 
	$m^{\ogi{k}} = h(n^{\ogi{k}}) = h(n^{\ogi{k+1}}) = m^{\ogi{k+1}}$. This completes
	the proof.
\qed \end{proof}

\begin{corollary}
	$\foi{n} \subsetneq \foi{n+1}$.
\end{corollary}
\begin{proof}
	By Theorem~\ref{thm:fo-inf-n}, the syntactic \circalgebra\ \monoid\ of any
	$\foi{n}$-definable "language" "divides" 
	a "block product" of copies of $\ali{n}$. By Lemma~\ref{main-lemma} and the fact
	that $\gnl{\ali{k}}=n$, $\gnl{\monoid} \leq n$. Note that, 
	$\ali{n+1}$ is the syntactic \circalgebra\ for the "language" $L$ defined by the $\foi{n+1}$ formula $\ei {n+1} x\ a(x)$.
	As $\gnl{\ali{n+1}}=n+1$, it follows that $L$ cannot be defined in $\foi{n}$.
\qed \end{proof}
\begin{theorem} There is no finite basis for a "block product" based characterization
for any of these logical systems $\foinf, \focut, \focut \cap \wmso$.
\end{theorem} 
\begin{proof}
Fix one of the logics ${\mathcal L}$ mentioned in the statement of the theorem.
It follows from Theorem~\ref{thm:foinf-focut-wmso} and the algberaic chacterization \cite{colcombetCLO}
of $\focut$ that the syntactic \circalgebras\ of ${\mathcal L}$-definable languages
are aperiodic.
Now suppose, for contradiction, let ${\mathcal L}$ admit a finite basis 
$B$ of  aperiodic \circalgebras\ for its  "block product" based chacterization. 
Since $B$ is finite,  there exists $n \in \mathbb{N}$ such that for all 
\circalgebras\ \monoid in $B$, $\gnl{\monoid} \leq n$. It 
follows by Lemma~\ref{main-lemma} that 
the syntactic \circalgebra\ \monoidN\ of \emph{every} ${\mathcal L}$-definable language 
has the property $\gnl{\monoidN} \leq n$.

Now consider the language $L$ defined by the \foinf\ sentence $\phi=\ei {n+1} x ~a(x)$.  By Theorem~\ref{thm:foinf-focut-wmso}, $L$ is ${\mathcal L}$-definable. 
Hence, the "gap-nesting-length" of the syntactic \circalgebra\ $K$ of $L$ is less than or equal
to $n$. However, $K$ is simply $\ali{n+1}$ and $\gnl{\ali{m+1}}=n+1$. This leads to
a contradiction.
\qed \end{proof}